\documentclass{easychair}
\usepackage{etoolbox}
\newtoggle{paper}

\togglefalse{paper}    

\usepackage{bm}
\usepackage{xspace} 
\usepackage{bussproofs}
\usepackage{amssymb}
\usepackage{url}
\usepackage{booktabs}
\usepackage{stmaryrd}
\usepackage{tikz}
\usetikzlibrary{graphs}
\usetikzlibrary {shapes.geometric} 
\usetikzlibrary {arrows.meta,automata,positioning}
\usepackage{subcaption}
\usepackage{fancyvrb}
\usepackage{thmtools, thm-restate}
\usepackage{cleveref}
\usepackage{proof}
\usepackage{enumitem}
\usepackage{adjustbox}
\usepackage{version}

\hypersetup{
  colorlinks = true,
  citecolor=gray,
  linkcolor=linkcol,
}

\definecolor{light-gray}{gray}{0.95}

\theoremstyle{plain}
  \newtheorem{theorem}{Theorem}[section]
  
  \newtheorem{proposition}[theorem]{Proposition}
  \newtheorem{cor}[theorem]{Corollary}
  
\theoremstyle{definition}
  \newtheorem{definition}{Definition}[section]
  \newtheorem{example}{Example}[section]

\definecolor{linkcol}{HTML}{5388C8}

\iftoggle{paper}{
 \includeversion{paper}
 \excludeversion{report}
}{
 \excludeversion{paper}
 \includeversion{report}
}

\newcommand{\set}[1]{\{\,#1\,\}}
\newcommand{\sql}[1]{\texttt{#1}}

\newcommand{\str}[1]{\mathcal{#1}}
\newcommand{\conf}[1]{( #1 )}

\newcommand{\rem}[1]{}
\renewcommand{\rem}[1]{\textcolor{red}{[#1]}}

\newcommand{\mm}[1]{}
\renewcommand{\mm}[1]{[\textcolor{blue}{MM: #1}]}

\newcommand{\teq}{\approx}
\newcommand{\tneq}{\not\approx}

\newcommand{\ter}[1]{\mathcal{T}(#1)}

\newcommand{\sint}{\mathsf{Int}}
\newcommand{\sbool}{\mathsf{Bool}}

\newcommand{\ite}{\mathsf{ite}}

\newcommand{\ltrue}{\top} 

\newcommand{\Mo}{\mathbf{I}}

\newcommand{\Sc}{\mathsf{S}}

\newcommand{\Ac}{\mathsf{A}}

\newcommand{\Ec}{\mathsf{E}}

\newcommand{\Bc}{\mathsf{B}}

\newcommand{\card}[1]{\lvert #1 \rvert}

\newcommand{\define}[1]{\textsl{#1}}

\newcommand{\unknown}{\ensuremath{\mathsf{unknown}}\xspace}
\newcommand{\unsat}{\ensuremath{\mathsf{unsat}}\xspace}
\newcommand{\sat}{\ensuremath{\mathsf{sat}}\xspace}

\newcommand{\cvc}{{\small cvc5}\xspace}

\newcommand{\nia}{\tname{NIA}}
\newcommand{\E}{\tname{EL}}

\newcommand{\nth}{\T_\mathsf{\nia}}

\newcommand{\bool}{\ensuremath{\mathsf{Bool}}\xspace}

\newcommand{\I}{\mathcal{I}}

\renewcommand{\vec}[1]{\bar{#1}}
\newcommand{\tname}[1]{\mathsf{#1}}

\newcommand{\ent}[1][]{\models_{#1}}

\newcommand{\arithmodel}{\mathcal{A}}

\newcommand{\bagEq}[1]{\equiv_{#1}}

\newcommand{\thsym}[1]{\ensuremath{\mathsf{#1}}\xspace}
\newcommand{\ms}{\thsym{Bag}}
\newcommand{\s}{\thsym{Set}}
\renewcommand{\t}{\thsym{Table}}
\renewcommand{\r}{\thsym{Relation}}
\newcommand{\tuple}{\thsym{Tuple}}
\renewcommand{\int}{\thsym{Int}}

\newcommand{\stringSort}{\thsym{String}}

\newcommand{\map}{\pi\xspace}

\newcommand{\bagFilter}{\thsym{bag.filter}}

\newcommand{\eleSort}{\ensuremath{\varepsilon}\xspace}

\newcommand{\sqin}{%
  \mathrel{\vphantom{\sqsubset}\text{%
      \mathsurround=0pt
      \ooalign{$\sqsubset$\cr$-$\cr}%
    }}%
}

\newcommand{\dsetminus}{\mathbin{{\setminus}\mspace{-5mu}{\setminus}}}
\newcommand{\setof}{\mathsf{setof}\xspace}

\newcommand{\msym}{\mathsf{m}}
\newcommand{\m}[2]{\msym(#1, #2)}
\newcommand{\filter}{\sigma}

\newcommand{\bag}{\mathsf{bag}}

\newcommand{\product}{\otimes}

\newcommand\squplus{\mathbin{\ooalign{$\sqcup$\cr%
      \hfil\raise0.42ex\hbox{$\scriptscriptstyle +$}\hfil\cr}}}

\newcommand{\bempty}[1]{\varnothing_{#1}}

\newcommand{\rewritesto}{~\longrightarrow~}

\newif \ifFANCYOPER
\FANCYOPERfalse

\newcommand{\vars}{\mathrm{Vars}}
\newcommand{\arity}{\mathrm{ar}}
\newcommand{\tup}[1]{\langle#1\rangle}
\newcommand{\seq}[2]{\bar{#1}_{#2}}

\newcommand{\symFont}[1]{\ensuremath{\mathsf{#1}}}

\newcommand{\tupleSort}{\symFont{Tup}}

\newcommand{\opemptyset}{[\,]}
\newcommand{\sempty}{[\,]}

\newcommand{\opprod}{\product}

\newcommand{\opunion}{\sqcup}
\newcommand{\opinter}{\sqcap}
\newcommand{\opsetminus}{\setminus}
\newcommand{\opin}{\mathrel{\ooalign{$\sqsubset$\cr{$-$}}}}
\newcommand{\opnotin}{\not\opin}
\newcommand{\opequal}{\approx}

\newcommand{\opsingleton}[1]{[\,#1\,]}

\newcommand{\Scclosed}{\Sc^*}

\newcommand{\termsof}[1]{\mathcal{T}(#1)}

\newcommand{\cadd}[2]{#1,\, #2}

\newcommand{\rn}[1]{\textsc{\small#1}\xspace}

\newcommand{\ruleInterDown}{\rn{Inter Down}}
\newcommand{\ruleInterUp}{\rn{Inter Up}}

\newcommand{\ruleUnionDown}{\rn{Union Down}}
\newcommand{\ruleUnionUp}{\rn{Union Up}}

\newcommand{\ruleDifferenceDown}{\rn{Diff Down}}
\newcommand{\ruleDifferenceUp}{\rn{Diff Up}}

\newcommand{\ruleProductUp}{{\rn{Prod Up}}\xspace}
\newcommand{\ruleProductDown}{{\rn{Prod Down}}\xspace}

\newcommand{\ruleJoinDown}{{\rn{Join Down}}\xspace}

\newcommand{\ruleSingleUp}{{\rn{Single Up}}}
\newcommand{\ruleSingleDown}{\rn{Single Down}}

\newcommand{\ruleSetDiseq}{\rn{Set Diseq}}
\newcommand{\ruleEqUnsat}{\rn{Eq Unsat}}

\newcommand{\ruleEmptyUnsat}{\rn{Empty Unsat}}

\newcommand{\nullable}{\mathsf{Nullable}}
\newcommand{\T}{\mathbb{T}}

\newcommand{\nullN}{\mathsf{null}}
\newcommand{\someN}{\mathsf{some}}
\newcommand{\valN}{\mathsf{val}}
\newcommand{\isSome}{\mathsf{isSome}}
\newcommand{\isNull}{\mathsf{isNull}}
\newcommand{\lift}{\mathsf{lift}}
\newcommand{\liaStar}{\(\text{QFPA}^\star\)\xspace}
\newcommand{\tablesTheory}{\T_\mathsf{Tab}}
\newcommand{\tablesSig}{\Sigma_\mathsf{Tab}}
\newcommand{\relationsTheory}{\T_\mathsf{Rel}}

\newcommand{\ver}{\; \vert \;}
\newcommand{\tSelect}{\mathsf{select}}
\newcommand{\tupleProject}{\mathsf{tuple.proj}}
\newcommand{\tableProject}{\mathsf{table.proj}}
\newcommand{\relationProject}{\mathsf{rel.proj}}
\newcommand{\select}{\mathsf{select}}

\newcommand{\project}{\mathsf{proj}}
\newcommand{\tabjoin}{\bowtie}

\newcommand{\ruleAConf}{\rn{A-Conf}}
\newcommand{\ruleBConf}{\rn{B-Conf}}
\newcommand{\ruleEConf}{\rn{E-Conf}}
\newcommand{\ruleBAProp}{\rn{B-A-Prop}}
\newcommand{\ruleAProp}{\rn{A-Prop}}
\newcommand{\ruleEProp}{\rn{E-Ident}}
\newcommand{\ruleBEProp}{\rn{B-E-Prop}}
\newcommand{\ruleBagDisequality}{\rn{Diseq}} 
\newcommand{\ruleNonNegative}{\rn{Nonneg}} 
\newcommand{\ruleBagEmpty}{\rn{Empty}}

\newcommand{\ruleBagConstructorOne}{\rn{Cons1}} 
\newcommand{\ruleBagConstructorTwo}{\rn{Cons2}} 
\newcommand{\ruleDisjointUnion}{\rn{Disj Union}} 
\newcommand{\ruleMaxUnion}{\rn{Max Union}}
\newcommand{\ruleBagIntersection}{\rn{Inter}} 
\newcommand{\ruleDifferenceSubtract}{\rn{Diff Sub}} 
\newcommand{\ruleDifferenceRemove}{\rn{Diff Rem}} 
\newcommand{\ruleSetof}{\rn{Setof}} 
\newcommand{\ruleTableProductUp}{\rn{Prod Up}} 
\newcommand{\ruleTableProductDown}{\rn{Prod Down}} 
\newcommand{\ruleTableJoinUp}{\rn{Join Up}}
\newcommand{\ruleTableJoinDown}{\rn{Join Down}}
\newcommand{\ruleRelationJoinUp}{\rn{Join Up}}
\newcommand{\ruleRelationJoinDown}{\rn{Join Down}}
\newcommand{\ruleBagFilterUp}{\rn{Filter Up}}
\newcommand{\ruleBagFilterDown}{\rn{Filter Down}}
\newcommand{\ruleBagMapUp}{\rn{Map Up}}
\newcommand{\ruleBagMapDownInjective}{\rn{Inj Map Down}} 
\newcommand{\ruleBagMapUpNonInjectiveUp}{\rn{NotInj Up}} 
\newcommand{\ruleBagMapUpNonInjectiveDown}{\rn{NotInj Down}} 
\newcommand{\ruleSetFilterUp}{\rn{Filter Up}}
\newcommand{\ruleSetFilterDown}{\rn{Filter Down}}
\newcommand{\ruleSetMapUp}{\rn{Map Up}}
\newcommand{\ruleSetMapDown}{\rn{Map Down}}
\newcommand{\inferR}[3][2]{\infer[#1]{#3}{#2}}

\newcommand{\calcite}{calcite\xspace}

\newcommand{\postgres}{PostgreSQL\xspace}
\newcommand{\sqlSolver}{SQLSolver\xspace}
\newcommand{\spes}{SPES\xspace}
\newcommand{\terminatingBagCount}{17}
\newcommand{\terminatingSetCount}{15}

\newcommand{\elementIndex}{\mathit{ind}}
\newcommand{\delem}{\mathit{delem}}
\newcommand{\mapSum}{\mathit{sum}}

\begin{document}
\title{Verifying SQL Queries \\ using Theories of Tables and Relations}
%
%
\author{Mudathir Mohamed\inst{1} \and 
  Andrew Reynolds\inst{1} \and 
  Cesare Tinelli\inst{1} \and 
  Clark Barrett\inst{2}} 
\institute{
  The University of Iowa
  \and
  Stanford University
}
\authorrunning{Mohamed et al.}
\titlerunning{Verifying SQL queries}

\maketitle              
%

\begin{abstract}
  We present a number of first- and second-order extensions to SMT theories
  specifically aimed at representing and analyzing SQL queries with
  join, projection, and selection operations.
  We support reasoning about SQL queries with either bag or set semantics for database tables.
  We provide the former via an extension of a theory of finite bags
  and the latter via an extension of the theory of finite relations.
  Furthermore, we add the ability to reason about tables with null values
  by introducing a theory of nullable sorts based on an extension of the theory of algebraic datatypes.
  We implemented solvers for these theories in the SMT solver cvc5 and
  evaluated them on a set of benchmarks derived from public sets of SQL equivalence problems.
  %
  %
  %
\end{abstract}
%
%
\section{Introduction}\label{sec:introduction}
%

%
The structured query language (SQL) is the dominant declarative query language in relational databases.
Two queries are equivalent in SQL if and only if they return the same table
for every database instance of the same schema.
Query equivalence problems are undecidable in general~\cite{foundationsOfDatabases}.
For \emph{conjunctive} queries, the problem is NP-complete 
under set semantics~\cite{cqNP}
and \(\prod_2^p\)-hard under bag semantics~\cite{pi2Hard}.
%
SQL query equivalence problems have many applications in
databases and software development,
including query optimization and sharing sub-queries in cloud databases.
There is a financial incentive to reduce the cost of these subqueries,
since cloud databases charge for data storage, network usage, and computation.

Recently, these problems got some attention from researchers
in formal verification
who have developed software tools to prove query equivalence 
in some SQL fragments.
To our knowledge, the state of the art of these tools is currently represented 
by \sqlSolver~\cite{sqlSolver}
which supports a large subset of SQL, and 
SPES, an earlier tool that was used to verify queries 
from cloud-scale applications~\cite{spes}.

We present an alternative solution for the analysis of SQL queries
based on a reduction to constraints in a new SMT theory of tables
with bag (i.e., multiset) semantics.
This work includes the definition of the theory and 
the development of a specialized subsolver for it
within the \cvc SMT solver.
We have extended a previous theory of finite bags~\cite{bagsPaper2008} 
with map and filter operators which are needed to support SQL \textit{select}
and \textit{where} clauses, respectively.
We represent table rows as tuples and define tables as bags of tuples.
We also support the \define{product} operator over tables.
While multiset semantics captures faithfully the way tables are treated 
in relational database management systems,
there is a lot of work in the database literature that is based on set semantics.
We provide set semantics as an alternative encoding based on a theory of finite relations 
by Meng et al.~\cite{relationsPaper}, extended in this case too to accommodate
SQL operations.

%
An initial experimental comparison of our implementation with \sqlSolver and SPES
places it between the two in terms of performance and supported features.
%
While there are several opportunities for further performance improvements, 
our solution has two main advantages with respect to previous work: 
$(i)$ it is not limited to SQL equivalence problems, and 
$(ii)$ it comes fully integrated in a state-of-the-art SMT solver
with a rich set of background theories.
Additionally, it could be further extended to provide support
for SQL queries combining set and bag-set semantics~\cite{Cohen06}. 
This opens up the door to other kinds of SQL query analyses
(including, for instance, query containment and query emptiness problems)
over a large set of types for query columns
(various types of numerical values, strings, enumerations, and so on).

%
\paragraph{\bf Specific Contributions}
We introduce a theory of finite tables by extending a theory of bags
with support for product, filter and map operators.
We also extend a theory of finite relations with map, filter, and inner join operators.
We introduce a theory of nullable sorts as an extension of a theory of algebraic datatypes.
These new theories enable the encoding in SMT of a large fragment of SQL 
under either multiset or set semantics
and the automated analysis of problems such as query equivalence.
We extend the \cvc SMT solver~\cite{DBLP:conf/tacas/BarbosaBBKLMMMN22}
with support for quantifier-free constraints
over any combination of the theories above and those already defined in \cvc.
We discuss an initial experimental evaluation on query equivalence benchmarks.

Our contribution does not include support for aggregations in SQL yet
but we plan to add that in future work.

\subsection{Related work}
A decision procedure for quantifier-free formulas in the theory of bags (QFB), 
or multisets, based on a reduction to quantifier-free Presburger arithmetic (QFPA) 
first appeared in Zarba~\cite{bagsZarba}.
The theory signature did not include cardinality constraints or difference operators.
These are supported in a new decision procedure 
by Logozzo et al.~\cite{bagsPaper2008}.
The new decision procedure reduces QFB to \liaStar, which extends QFPA with formulas
\(\vec{u} \in \{\vec{x} \; \vert \; \phi\}\), where \(\phi \in \) QFPA. 
Then, the \liaStar formula is translated into a QFPA formula, but 
with the addition of space overhead \cite{liaStarApproximation}.
An improved decision procedure which addresses the space overhead issue
by using approximations and interpolation with a set of Constrained Horn Clauses
was provided in Levatich et al.~\cite{liaStarApproximation}.
Our work is closest to Zarba's~\cite{bagsZarba},
with additional support for map and filter operators.
As in that work, we do not support the cardinality operator yet.\footnote{%
Indirect support for that will be provided in future work
through the support for SQL aggregations.
}
%
%
%
\cvc already supports the theory of finite sets~\cite{DBLP:journals/lmcs/BansalBRT18} and its extension to finite relations~\cite{relationsPaper}.
We add support for the map and filter operators to \cvc's theory solver for sets and,
by extension, to the theory solver for relations,
proving the decidability of the satisfiability problem
in a restricted fragment of the theory of finite sets.
That fragment is enough to handle the benchmarks considered in our experiments
under set semantics.

%
Cosette is an automated tool specifically written to prove 
SQL query equivalence~\cite{cosette}.
%
%
%
To do that, it translates the two SQL queries into algebraic expressions
over an unbounded semiring,
which it then normalizes to a \emph{sum-product} normal form.
Finally, it searches for an isomorphism between the two normal forms using
a custom decision procedure.
If an isomorphism is found, the two queries are declared equivalent.
To show that queries are inequivalent, Cosette translates the SQL queries into bounded lists
and uses an SMT solver to find a counterexample to the equivalence.
Cosette is a significant step forward in checking SQL query equivalences.
However, it has several limitations.
For instance, it does not support null values, intersection, difference, arithmetic operations, or string operations, which are all common in SQL queries.


EQUITAS ~\cite{equitas} and its successor SPES~\cite{spes}
are used to identify shared subqueries automatically in cloud databases.
Both support aggregate queries and null values. 
They use symbolic representations to prove query equivalence.
EQUITAS follows set semantics, whereas SPES follows bag semantics for tables.
EQUITAS starts by assigning symbolic tuples for the input tables in the queries,
and then applies specialized algorithms to build two formulas representing 
output tuples for the two queries.
If the two formulas are equivalent, then the two queries are classified as such.
SPES also uses  symbolic representations for queries.
However, it reduces the query equivalence problem to the existence
of an identity map between the tuples returned by the two queries~\cite{spes}.
In experimental evaluations~\cite{spes}, SPES proved more queries than EQUITAS and was 3 times as fast.
Both tools have the limitation of only supporting queries 
with similar structure.
They do not process queries with 
structurally different abstract syntax trees
(e.g., queries with different number of joins),
or queries that use basic operations such as difference or intersection.
They also do not support concrete tables, built using the keyword \verb|VALUES|.

\sqlSolver, which to our knowledge represents the current state of the art, 
was released recently and addresses many of the limitations highlighted above~\cite{sqlSolver}.
It follows bag semantics and, similar to Cosette, reduces the input queries
to unbounded semiring expressions.
However, it then translates them into formulas in an extension 
of \liaStar that supports nested, parametrized, or nonlinear summation.
This extension supports projection, product, and aggregate functions.
\sqlSolver implements algorithms similar to those 
in Levatich et al.~\cite{liaStarApproximation} to translate these formulas into QFPA.
\sqlSolver dominates other tools both in terms of performance and 
expressiveness of the supported SQL fragment.
However, it lacks the ability to generate counterexamples 
for inequivalent queries~\cite{sqlSolver}.
%

We follow a different approach from all the equivalence checkers discussed above.
Our solver, incorporated into \cvc, supports difference and intersection operations,
as well as evaluation on concrete tables.
Thanks to the rich set of background theories provided by \cvc,
it also supports arithmetic and string operations, as well as null values.
%
%
%
Finally, our solver is not restricted to SQL query equivalence,
as it supports in general any quantifier-free statements over SQL queries.
As a consequence, it can also be used for other applications such as, for instance,
checking for query containment or emptiness~\cite{qex1}.



\subsection{Formal Preliminaries}\label{sec:preliminaries}
We define our theories and our calculi in the context of many-sorted logic 
with equality and polymorphic sorts and functions.
%
We assume the reader is familiar with the following notions from that logic:
signature, term, 
formula, free variable,
interpretation, and satisfiability of a formula in an interpretation.
%
%
Let \(\Sigma\) be a many-sorted signature.
We will denote sort parameters in polymorphic sorts with \(\alpha\) and \(\beta\),
and denote monomorphic sorts with \(\tau\).
We will use \(\teq\) as the (infix) logical symbol for equality --- which
has polymorphic rank \(\alpha \times \alpha\) and is always interpreted 
as the identity relation over \(\alpha\).
We assume all signatures \(\Sigma\) contain the Boolean sort \(\sbool\),
always interpreted as the binary set \(\{\mathit{true},\mathit{false}\}\),
and two Boolean constant symbols, \(\ltrue\) and \(\bot\), 
for \(\mathit{true}\) and \(\mathit{false}\).
Without loss of generality, we assume \(\teq\) is the only predicate symbol
in \(\Sigma\), as all other predicates can be modeled as functions
with return sort \(\sbool\).
We will write, e.g., \(p( x )\) as shorthand for \(p( x ) \teq \ltrue\),
where \(p( x )\) has sort \(\sbool\).
We write \(s \tneq t \) as an abbreviation for \(\lnot\,s \teq t\).
\begin{report}
When talking about sorts, we will normally say just ``sort'' to mean monomorphic sort,
and say ``polymorphic sort'' otherwise.
\end{report}

%
%

A \define{\(\Sigma\)-term} is a well-sorted term, all of whose function symbols 
are from \(\Sigma\).
A \define{\(\Sigma\)-formula} is defined analogously.
If \(\varphi\) is a \(\Sigma\)-formula and \(\I\) a \(\Sigma\)-interpretation,
we write \(\I \models \varphi\) if \(\I\) satisfies \(\varphi\).
If \(t\) is a term, we denote by \(\I(t)\) the value of \(t\) in \(\I\).
%
%
A \define{theory} is a pair \(\T = (\Sigma, \Mo)\), where
\(\Sigma\) is a signature and  \(\Mo\) is a class of \(\Sigma\)-interpretations
that is closed under variable reassignment
(i.e., every \(\Sigma\)-interpretation that differs from one in \(\Mo\)
only in how it interprets the variables is also in \(\Mo\)).
\(\Mo\) is also referred to as the \define{models} of \(\T\).
A \(\Sigma\)-formula \(\varphi\) is
\define{satisfiable} (resp., \define{unsatisfiable}) in \(\T\)
if it is satisfied by some (resp., no) interpretation in \(\Mo\).
A set \(\Gamma\) of \(\Sigma\)-formulas \define{entails} in \(\T\)
a \(\Sigma\)-formula \(\varphi\), written \(\Gamma \ent[\T] \varphi\),
if every interpretation in \(\Mo\) that satisfies all formulas in \(\Gamma\)
satisfies \(\varphi\) as well.
%
%
We write \(\Gamma \ent \varphi\) to denote that
\(\Gamma\) entails \(\varphi\) in the class of all \(\Sigma\)-interpretations.
Two $\Sigma$-formulas are \define{equisatisfiable in \(\T\)}
if for every interpretation \(\str A\) of \(\T\) that satisfies one, there is
an interpretation
of \(\T\) that satisfies the other and differs from \(\str A\) at most
in how it interprets the free variables not shared by the two formulas.
%

%

\section{Theory of Tables}\label{sec:tables_theory}


\begin{figure}[tbp]
  \scriptsize
\centering
  \(
    \begin{array}{l}
      \begin{array}{@{}l@{\quad}l@{\quad}l@{\quad}l@{}}
        \toprule
        \textbf{Symbol} & \textbf{Type}                                                             & \textbf{SMT-LIB syntax}                & \textbf{Description}                          \\
        \midrule
        n               & \sint                                                                      & \text{n}                       & \text{All constants } \verb|n| \in \mathbb{N} \\
        +               & \sint \times \sint \to \sint                                               & \verb|+|                       & \text{Integer addition}                       \\
        *               & \sint \times \sint \to \sint                                               & \verb|*|                       & \text{Integer multiplication}                 \\
        -               & \sint \to \sint                                                            & \verb|-|                       & \text{Unary Integer minus}                    \\
        {\leq}          & \sint \times \sint \rightarrow \sbool                                      & \verb|<=|                      & \text{Integer inequality}                     \\
        \midrule        
        \bempty{\alpha}         & \ms(\alpha)                                                                 & \verb|bag.empty|               & \text{Empty bag}                              \\
        \bag             & \alpha \times \int \rightarrow \ms(\alpha)                                   & \verb|bag|                     & \text{Bag constructor}                        \\
        \msym           & \alpha \times \ms(\alpha) \rightarrow \int                                   & \verb|bag.count|               & \text{Multiplicity}                           \\
        \setof          & \ms(\alpha) \rightarrow \ms(\alpha)                                          & \verb|bag.setof|                & \text{Duplicate remove }                     \\
        \sqcup          & \ms(\alpha) \times \ms(\alpha) \rightarrow \ms(\alpha)                        & \verb|bag.union_max|           & \text{Max union}                              \\
        \squplus        & \ms(\alpha) \times \ms(\alpha) \rightarrow \ms(\alpha)                        & \verb|bag.union_disjoint|      & \text{Disjoint union}                         \\
        \sqcap          & \ms(\alpha) \times \ms(\alpha) \rightarrow \ms(\alpha)                        & \verb|bag.inter_min|           & \text{Intersection}                           \\
        \setminus       & \ms(\alpha) \times \ms(\alpha) \rightarrow \ms(\alpha)                        & \verb|bag.diff_subtract| & \text{Difference subtract}                    \\
        \dsetminus      & \ms(\alpha) \times \ms(\alpha) \rightarrow \ms(\alpha)                        & \verb|bag.diff_remove|   & \text{Difference remove}                      \\
        \sqin           & \alpha \times \ms(\alpha) \rightarrow \bool                                  & \verb|bag.member|              & \text{Member }                                \\
        \sqsubseteq     & \ms(\alpha) \times \ms(\alpha) \rightarrow \bool                             & \verb|bag.subbag|              & \text{Subbag}                                 \\
        \midrule
        %
        \filter          & 
        \left(\alpha \rightarrow \bool\right) \times \ms(\alpha) \rightarrow \ms(\alpha)                & \verb|bag.filter|              & \text{Bag filter}                             \\
        \map             & 
        \left(\alpha_1 \rightarrow \alpha_2\right) \times \ms(\alpha_1) \rightarrow \ms(\alpha_2)        & \verb|bag.map|                 & \text{Bag map}                                \\
        \midrule
        \tup{ \ldots} &
        \alpha_0 \times \dots \times \alpha_k \rightarrow \tuple(\alpha_0, \ldots, \alpha_k)                  & \verb|tuple|                    & \text{Tuple constructor}                      \\ 
        
        \select_i          & \tuple(\alpha_0, \ldots, \alpha_k) \rightarrow \alpha_i        & \verb|(_ tuple.select i)|    & \text{Tuple selector}                         
        \\[.8ex]
        \tupleProject_{i_1 \ldots i_n}        & 
        \tuple(\alpha_0, \ldots, \alpha_k) \rightarrow \tuple(\alpha_{i_1}, \ldots, \alpha_{i_n})
          & (\verb|_ tuple.proj | i_1 \cdots i_n)    & \text{Tuple projection}                         \\
        \midrule
        \product        & \t(\bm{\alpha}) \times \t(\bm{\beta}) \rightarrow \t(\bm{\alpha}, \bm{\beta}) & \verb|table.product|           & \text{Table cross join}                       \\
        \tabjoin_{i_1j_1 \cdots i_pj_p}      & 
          \t(\bm{\alpha}) \times \t(\bm{\beta}) \rightarrow \t(\bm{\alpha}, \bm{\beta}) 
        & (\verb|_ table.join | i_1j_1 \cdots i_pj_p)             & \text{Table inner join}
        \\[.8ex]
        \tableProject_{i_1 \ldots i_n}   &
          \t(\alpha_0, \ldots, \alpha_k) \rightarrow \t(\alpha_{i_1}, \ldots, \alpha_{i_n})
            & (\verb|_ table.proj | i_1 \cdots i_n)    & \text{Table projection}                         \\
        \bottomrule
      \end{array}
      \\[4ex]
    \end{array}
  \)
  \caption{
    Signature \(\tablesSig\) for the theory of tables.
    Here \(\t(\bm{\alpha}, \bm{\beta})\) is a shorthand for \(\t(\alpha_0, \ldots, \alpha_p, \beta_0, \ldots, \beta_q)\)
    when \(\bm{\alpha} = \alpha_0, \ldots, \alpha_p\) and \(\bm{\beta} = \beta_0, \ldots, \beta_q\).
  }
  \label{fig:tables_sig}
\end{figure}

We define a many-sorted theory \(\tablesTheory\) of (database) tables.
Its signature \(\tablesSig\) is given in Figure~\ref{fig:tables_sig}.
%
We use \(\alpha\) and \(\beta\), possibly with subscripts, as sort parameters
in polymorphic sorts.
%
%
%
The theory includes the integer sort \int and a number of integer operators,
with the same interpretation as in the theory of arithmetic.
Additionally, \(\tablesTheory\) has three classes of sorts,
with a corresponding polymorphic sort constructor:
function sorts, tuple sorts, and bag sorts.
\define{Function sorts} are monomorphic instances of 
\(\alpha_1 \times \dots \times \alpha_k \rightarrow \alpha\) for all \(k \geq 0\).
%
%
\define{Tuple sorts} are constructed by the varyadic constructor \(\tuple\) 
which takes zero or more sort arguments.
With no arguments, \(\tuple\) denotes the singleton set containing the empty tuple.
With \(k + 1\) arguments for \(k \geq 0\),
\(\tuple (\tau_0, \ldots, \tau_k)\) denotes the set of tuples of size \(k+1\) 
with elements of sort \(\tau_0, \ldots, \tau_k\), respectively.
\define{Bag sorts} are monomorphic instances \(\ms(\tau)\) of \(\ms(\alpha)\).
The sort \(\ms(\tau)\) denotes the set of all \emph{finite} bags 
(i.e., finite multisets) of elements of sort \(\tau\).
%
%
We model tables as bags of tuples.
We write \(\t(\tau_0, \ldots, \tau_k)\) as shorthand for the sort
\(\ms(\tuple (\tau_0, \ldots, \tau_k))\).
The sort \t, with no arguments, abbreviates \(\ms(\tuple)\).\footnote{%
And so denotes the set of all tables
containing just occurrences of the empty tuple.
}
Following databases terminology,
we refer to \(\tau_0, \dots, \tau_k\) as the \define{columns} of 
\(\t(\tau_0, \ldots, \tau_k)\),
and to the elements of a given table as its \define{rows}.
For convenience, we index the columns of a table by natural numbers
(starting with 0), instead of alphanumeric names, as in SQL.

The symbols in the first five lines in Figure~\ref{fig:tables_sig} are 
from arithmetic and are interpreted as expected.
The next eleven function symbols operate on bags and are defined
as in Logozzo~et al.~\cite{bagsPaper2008}.
Specifically, for all sorts \(\tau\),
\(\bempty{\tau}\) is interpreted as the empty bag of elements of sort \(\tau\).
The term \(\bag(e, n) \) denotes a singleton bag with \(n\) occurrences
of the bag element \(e\) if \(n \geq 1\);
otherwise; it denotes \(\bempty{\tau}\) where \(\tau\) is the sort of \(e\).
The term \(\m{e}{s}\) denotes the \define{multiplicity} of element \(e\)
in bag \(s\), that is, the number of times \(e\) occurs in \(s\).
Its codomain is the set of natural numbers.
For convenience, we use \int as the codomain and, during reasoning,
assert \(\m{e}{s}\geq 0\) for each multiplicity term.
The term \(\setof(s)\) denotes the bag with the same elements as \(s\)
but without duplicates.
The predicate \(e \sqin s\) holds iff element \(e\) has
positive multiplicity in bag \(s\).
The predicate \(s \sqsubseteq t\) holds iff bag \(s\) is contained in bag \(t\)
in the sense that \(\m{e}{s}\leq \m{e}{t}\) for all elements \(e\).
The binary operators \(\sqcup,\squplus, \sqcap, \setminus, \dsetminus \) are
interpreted as functions that take two bags \(s\) and \(t\) and return
their max union, disjoint union, subtract difference, and remove difference,
respectively, making the following equalities valid in \(\tablesTheory\):
\[
  \begin{array}{rcl@{\qquad}rcl}
    \m{e}{s \sqcup t}     & \teq & \max(\m{e}{s}, \m{e}{t})           &
    \m{e}{s \squplus t}   & \teq & \m{e}{s} + \m{e}{t}                  \\
    \m{e}{s \sqcap t}     & \teq & \min(\m{e}{s}, \m{e}{t})           &
    \m{e}{s \setminus t}  & \teq & \max(\m{e}{s} - \m{e}{t}, 0)         \\
    \m{e}{s \dsetminus t} & \teq & \ite(\m{e}{t} \geq 1, 0, \m{e}{s}) &
  \end{array}
\]

The next two symbols in Figure~\ref{fig:tables_sig}
are 
the filter and map functionals.
These symbols require an SMT solver that supports higher-order logic,
which is the case for cvc5~\cite{ho}.
%
%
The term \(\filter(p, s)\) denotes the bag consisting of the elements
of bag \(s\) that satisfy predicate \(p\), with the same multiplicity
they have in \(s\).
%
%
The term \(\map(f, s)\) denotes the bag consisting of the elements \(f(e)\),
for all \(e\) that occur in \(s\).
The multiplicity of \(f(e)\) in \(\filter(f, s)\) is
the sum of the multiplicities (in \(s\)) of all the elements of \(s\)
that \(f\) maps to \(f(e)\).
Note that while \(\map(f,s)\) and \(s\) have the same cardinality,
i.e., the same number of element occurrences,
\(\m{f(e)}{\map(f,s)}\) may be greater than \(\m{e}{s}\) for some elements \(e\)
unless \(f\) is injective.

The last six symbols denote dependent families of functions
over tuples and tables.
The term \(\tup{e_1, \dots, e_k}\)
is interpreted as the tuple comprised of the elements \(e_1, \dots, e_k\),
in that order, with \(\tup{}\) denoting the empty tuple.
%
For \(0 \leq i \leq k\) where \(k+1\) is the size of a tuple \(t\),
\(\select_i(t)\) is interpreted as the element at position $i$ of \(t\).
%
Note that \(i\) in \(\select_i\) is a numeral, not a symbolic index.
This means in particular that \(\select_3(\tup{a,b})\) is an ill-sorted term.
\(\tupleProject\) takes an unbounded number
of integer arguments, followed by a tuple.
\(\tupleProject_{i_1 \ldots i_n}\) where $n \geq 1$ and  
each \(i_j\) is an element of \(\set{0, \ldots, k}\),
applies to any tuple of size at least \(k+1\)
and returns the tuple obtained by collecting the values 
at position \(i_1, \ldots, i_n\) in \(t\).
In other words, it is equivalent to
\( \tup{\select_{i_1}(t), \ldots, \select_{i_n}(t)} \).
Note that \(i_1, \ldots, i_n\) are not required to be distinct.
When \(n = 0\), the term is equivalent to the empty tuple.
\(\tableProject_{i_1 \ldots i_n}\) extends the notion of projection to tables.
It is similar to \(\tupleProject\) except that it takes a table 
instead of a tuple as argument.
%
Note that the cardinality of \(\tableProject_{i_1 \ldots i_n}(s)\) is the same
as the cardinality of \(s\).
%
The term \(t \product t'\) is interpreted as the \define{cross join}
of tables \(t\) and \(t'\),
with every tuple occurrence in \(t\) being concatenated
with every tuple occurrence in \(t'\).
The operator \(\tabjoin_{i_1j_1 \cdots i_nj_n}\) is indexed by \(n\) pairs of natural numbers.
It takes two tables as input, each with at least \(n\) columns, 
and outputs the inner join of these tables on the columns specified 
by these index pairs. 
The paired columns have to be of the same sort.
%
%
%
%
Notice that if \(n = 0\), the join is equivalent to a product.
%

\paragraph{Simplifying Assumptions}
\emph{To simplify the exposition} and the description of the calculus, from now on,
we will consider only bags whose elements are not themselves bags, and
only tuples whose elements are neither tuples nor bags.
Note that this \emph{non-nestedness} restriction applies to tables as well ---
as they are just bags of tuples.
This is enough in principle to model and reason about SQL tables.\footnote{%
Commercial databases do allow table elements to be tuples. 
We could easily support this capability in the future simply by providing two kinds 
of tuple sorts, one for rows and one for table elements.
}
We stress, however, that none of these restrictions are necessary in our approach, 
nor required by our implementation, where we rely on \cvc's ability to reason modularly
about arbitrarily nested sorts.
Finally, we will not formalize in the calculus 
how we process constraints containing table projections 
(i.e., applications of \(\tableProject\)).
Such constraints are reduced internally to map constraints, with mapping functions
generated on the fly, and added to the relevant background solver.
\begin{report}
For instance, a constraint of the form \(s \teq \tableProject_{2,0}(t)\)
where \(t\) has sort \(\t(\tau_0,\ldots,\tau_3)\), say,
is reduced to the constraint
\(s \teq \pi(f,t)\) where \(f\) is a function symbol defined internally
to be equivalent to \(\lambda e.\, \tupleProject_{2,0}(e)\).
\end{report}

%

%

\begin{figure}[tbp]
  \centering
  \(
  \begin{array}{rl@{\quad\qquad}rl@{\quad\qquad}rl}
    \m{e}{\bempty{\eleSort}} & \rewritesto 0
                          & s \squplus \bempty{\eleSort}               & \rewritesto s                          
    \\
    e \sqin s             & \rewritesto 1 \leq \m{e}{s} & \bempty{\eleSort} \squplus t &\rewritesto t
    \\
    \bag(e, -n)           & \rewritesto \bempty{\eleSort}
                          & s \sqsubseteq t                         & \rewritesto (s \setminus t) \teq \bempty{\eleSort}
    \\
  \end{array}
  \)
  \caption{%
    Simplification rules for \(\tablesSig\)-terms.
    In the last two rules, \(\eleSort\) is the sort of \(e\) and of the elements of \(s\), respectively; \(n\) is a numeral.
  }
  \label{fig:rewriting}
\end{figure}

\begin{definition}\label{def:table_constraints}
A (monomorphic) sort is an \define{element sort} 
if it is not an instance of  \(\ms(\alpha)\).
An \define{element term} is a term of an \define{element sort}.
A \define{tuple/bag/table term} is a term of tuple/bag/table sort,
respectively.
A \define{\(\tablesTheory\)-atom} is an atomic \(\tablesSig\)-formula of the form
\(t_1 \approx t_2\),
\(e \sqin s\), or
\(s_1 \sqsubseteq s_2\),
where \(t_1\) and \(t_2\) are terms of the same sort,
\(e\) is a term of some element sort \(\tau\) and \(s\), \(s_1\), and \(s_2\)
are terms of sort \(\ms(\tau)\).

  A \(\tablesSig\)-formula \(\varphi\) is a \define{table constraint}
  if it has the form \(s \teq t\) or \(s \tneq t\);
  it is an \define{arithmetic constraint}
  if it has the form \(s \teq t\), \(s \tneq t\), or \(s \leq t\),
  where
  \(s\), \(t\) are terms of sort \int;
  it is an \define{element constraint}
  if it has the form
  \(e_1 \teq e_2, e_1 \tneq e_2,p(e),f(e_1) \teq e_2\), 
  where \(e, e_1, e_2\) are terms of some element sort,
  \(p\) is a function symbol of sort \(\eleSort \rightarrow \bool\)
  for some element sort \(\eleSort\), and
  \(f\) is a function symbol of sort \(\eleSort_1 \rightarrow \eleSort_2\),
  for some element sorts \(\eleSort_1\) and \(\eleSort_2\).
\end{definition}
Note that table constraints include (dis)equalities
between terms of any sort.
%
This implies that (dis)equalities between terms of sort \int are
both table and arithmetic constraints.


\subsection{Calculus}\label{sec:tables_calculus}
We now describe a tableaux-style calculus with derivation rules 
designed to determine the satisfiability in \(\tablesTheory\) 
of quantifier-free \(\tablesSig\)-formulas \(\varphi\).
To simplify the description, we will pretend that 
all tables have columns of the same element sort, denoted generically by \(\eleSort\).

Without loss of generality, we assume that the atoms 
of \(\varphi\) are in reduced form with respect to the (terminating) rewrite system 
in Figure~\ref{fig:rewriting},  which means that $\varphi$ is a Boolean combination 
of only equality constraints and arithmetic constraints.
%
Thanks to the following lemma, we will further focus on just sets
of table constraints and arithmetic constraints.\footnote{%
\begin{paper}
Proofs of this lemma and later results can be found in 
a longer version of this paper~\cite{arxiv}.
\end{paper}  
\begin{report}
Proofs of this lemma and later results can be found in the appendix.
\end{report}  
}

\begin{restatable}{lemma}{bagArithConstraints}
  \label{lem:bagArithConstraints}
  For every quantifier-free \(\tablesSig\)-formula \(\varphi\),
  there are sets \(B_1,\ldots,B_n\) of table constraints,
  sets \(A_1,\ldots,A_n\) of arithmetic constraints,
  and sets \(E_1,\ldots,E_n\) of elements constraints
  such that
  \(\varphi\) is satisfiable in \(\tablesTheory\)
  iff \(A_i \cup B_i \cup E_i\) is satisfiable in \(\tablesTheory\) for some \(i\in[1,n]\).
\end{restatable}

As a final simplification, we can also assume, without loss of generality, that 
for every term \(t\) of sort \(\tuple(\tau_0, \ldots, \tau_k)\) occurring 
in one of the sets \(B_i\) above,
\(B_i\) also contains the constraint \(t \teq \tup{x_0, \ldots, x_k}\)
where \(x_0, \ldots, x_k\) are variables of sort \(\tau_0, \ldots, \tau_k\),
respectively.

\paragraph{\bf Configurations and Derivation Trees.}
The calculus operates on data structures we call \emph{configurations}.
These are either the distinguished configuration \unsat or
triples \(\conf{\Ac,\Bc, \Ec}\) consisting of
a set \(\Ac\) of arithmetic constraints,
a set \(\Bc\) of table constraints,
and a set \(\Ec\) of element constraints.
Our calculus is a set of derivation rules that apply to configurations.

We assume we have a (possibly multi-theory) \define{element solver} 
that can decide the satisfiability of constraints in \(\Ec\).
This requires the computability of all predicates \(p\) and
functions \(f\) used as arguments in applications 
of filter (\(\sigma\)) and map (\(\pi\)), respectively.
We also define the set \(W\) to be an infinite set of fresh variables,
which will be used in specific derivation rules.

Derivation rules take a configuration and, if applicable to it, 
generate one or more alternative configurations.
A derivation rule \define{applies} to a configuration \(c\)
if all the conditions in the rule's premises hold for \(c\) \emph{and}
the rule application is not redundant.
An application of a rule is \define{redundant} if it has a conclusion
where each component in the derived configuration is a subset of
the corresponding component in the premise configuration.
%
%
%

A configuration other than \unsat is \define{saturated with respect to a set} 
\(R\) of derivation rules
if every possible application of a rule in \(R\) to it is redundant.
It is \define{saturated} if it is saturated with respect to all derivation rules
in the calculus.
A configuration \(\conf{\Ac,\Bc, \Ec}\) is \define{satisfiable} in \(\tablesTheory\)
if the set \(\Ac \cup \Bc \cup \Ec\) is satisfiable in \(\tablesTheory\).

A \define{derivation tree} is a (possibly infinite) tree 
where each node is a configuration whose (finitely-many) children, if any,
are obtained by a non-redundant application of a rule of the calculus
to the node.
A derivation tree is \define{closed} if it is finite and all its leaves are \unsat.
As we show later,
a closed derivation tree with root \(\conf{\Ac,\Bc, \Ec}\) is a proof that
\(\Ac \cup \Bc \cup \Ec\) is unsatisfiable in  \(\tablesTheory\).
In contrast, a derivation tree with root \(\conf{\Ac,\Bc, \Ec}\) and 
a saturated leaf with respect to all the rules of the calculus
is a witness that \(\Ac \cup \Bc \cup \Ec\) is satisfiable in \(\tablesTheory\).


\begin{figure}
 \centering
 \begin{prooftree}  
  \LeftLabel{\ruleAConf}
  \AxiomC{\(\Ac \models_{\nia} \bot \)}
  \UnaryInfC{\unsat}
  \DisplayProof \hskip 2em
  \LeftLabel{\ruleBConf}
  \AxiomC{\(t \tneq t \in \Bc^* \)}
  \UnaryInfC{\unsat}
  \DisplayProof \hskip 2em
  \LeftLabel{\ruleEConf}
  \AxiomC{\(\Ec \models_\E \bot \)}
  \UnaryInfC{\unsat}
 \end{prooftree}

 \begin{prooftree}  
  \LeftLabel{\ruleBAProp}
  \AxiomC{\(s \teq t \in \Bc^* \)}
  \AxiomC{\(s, t : \int \)}
  \BinaryInfC{\(\Ac := \Ac, s \teq t\)}
  \DisplayProof \hskip 1em
  \LeftLabel{\ruleBEProp}
  \AxiomC{\(e_1 \teq e_2 \in \Bc^*\)}
  \AxiomC{\(e_1, e_2\) are elem.~terms}
  \BinaryInfC{\(\Ec := \Ec, e_1 \teq e_2\)}
 \end{prooftree}
  
 \begin{prooftree}      
  \LeftLabel{\ruleEProp}
  \AxiomC{\(e_1,e_2 \in \ter{\Bc^*}\)}
  \AxiomC{\(e_1,e_2\) are element terms of the same sort}
  \BinaryInfC{
    \(\Bc := \Bc, e_1 \teq e_2
    ~~~\vert\vert~~~
    \Bc := \Bc, e_1 \tneq e_2\)}
 \end{prooftree}

 \begin{prooftree}      
  \LeftLabel{\ruleAProp}
  \AxiomC{\(\Ac \models_{\nia} s \teq t \)}
  \AxiomC{\(s , t \in \Ac \)}
  \AxiomC{\(s \text{ or } t\) is a multiplicity term}
  \TrinaryInfC{\(\Bc := \Bc, s \teq t\)}
 \end{prooftree}

 \begin{prooftree}  
  \LeftLabel{\ruleBagDisequality}
  \AxiomC{\(s \tneq t \in \Bc^* \)}
  \AxiomC{\(w\) is a fresh variable}
  \BinaryInfC{\(\Bc := \Bc, \m{w}{s} \tneq \m{w}{t} \quad \Ac := \Ac, \m{w}{s} \tneq \m{w}{t}
    \)}
  \DisplayProof \hskip 1em
  \LeftLabel{\ruleNonNegative}
  \AxiomC{\(\m{e}{s} \in \ter{\Bc^*}\)}
  \UnaryInfC{\(\Ac:= \Ac,0 \leq \m{e}{s}\)}
 \end{prooftree}

 \begin{prooftree}  
  \LeftLabel{\ruleBagConstructorOne}
  \AxiomC{\(s \teq \bag(e,n) \in \Bc^*\)}
  \AxiomC{\( n \leq 0 \notin \Ac \)}
  \AxiomC{\(1 \leq n \notin \Ac \)}
  \TrinaryInfC{\(
    \begin{array}{ll}
      & \Ac:= \Ac, n \leq 0, \m{e}{s} \teq 0 \quad
        \Bc := \Bc, s \teq \bempty{\eleSort} \\
      \vert\vert & 
      \Ac:= \Ac, 1 \leq n, \m{e}{s} \teq n \quad
      \Bc := \Bc, s \tneq \bempty{\eleSort}
    \end{array}
    \)}
 \end{prooftree}

 \begin{prooftree}  
  \LeftLabel{\ruleBagConstructorTwo}
  \AxiomC{\(s \teq \bag(e,n) \in \Bc^*\)}
  \AxiomC{\(x \tneq e \in \Bc^*\)}
  \BinaryInfC{\(
    \begin{array}{l}
      \Ac:= \Ac, \m{x}{s} \teq 0
    \end{array}
    \)
  }
  \DisplayProof \hskip 1em
  \LeftLabel{\ruleBagEmpty}
  \AxiomC{\(s \teq \bempty{\eleSort} \in \Bc^*\)}
  \AxiomC{\(\m{e}{s} \in \ter{\Bc}\)}
  \LeftLabel{\ruleBagEmpty}
  \BinaryInfC{\(\Ac:= \Ac,\m{e}{s} \teq 0\)}
 \end{prooftree}

 \begin{prooftree}  
  \LeftLabel{\ruleDisjointUnion}
  \AxiomC{\(s \teq t \squplus u \in \Bc^*\)}
  \AxiomC{\(\m{e}{v} \in \ter{\Bc}\)}
  \AxiomC{\(v \in \{s, t, u\}\)}
  \TrinaryInfC{\(
    \begin{array}{c}
      \Ac := \Ac,
      \m{e}{s} \teq \m{e}{t} + \m{e}{u}
    \end{array}
    \)}
 \end{prooftree}

 \begin{prooftree}  
  \LeftLabel{\ruleMaxUnion}
  \AxiomC{\( s \teq t \sqcup u \in \Bc^*\)}
  \AxiomC{\(\m{e}{v} \in \ter{\Bc}\)}
  \AxiomC{\(v \in \{s, t, u\}\)}
  \TrinaryInfC{\(
    \begin{array}{c}
      \Ac := \Ac,
      \m{e}{s} \teq \max(\m{e}{t}, \m{e}{u})
    \end{array}
    \)}
 \end{prooftree}

 \begin{prooftree}  
  \LeftLabel{\ruleBagIntersection}
  \AxiomC{\( s \teq t \sqcap u \in \Bc^*\)}
  \AxiomC{\(\m{e}{v} \in \ter{\Bc}\)}
  \AxiomC{\(v \in \{s, t, u\}\)}
  \TrinaryInfC{\(
    \begin{array}{c}
      \Ac := \Ac,
      \m{e}{s} \teq \min(\m{e}{t}, \m{e}{u})
    \end{array}
    \)}
 \end{prooftree}

 \begin{prooftree}  
  \LeftLabel{\ruleDifferenceSubtract}
  \AxiomC{\( s \teq t \setminus u  \in \Bc^*\)}
  \AxiomC{\(\m{e}{v} \in \ter{\Bc}\)}
  \AxiomC{\(v \in \{s, t, u\}\)}
  \TrinaryInfC{\(
    \begin{array}{ll}
                 & \Ac := \Ac,\m{e}{t} \leq \m{e}{u}, \m{e}{s} \teq 0 \\
      \vert\vert & \Ac := \Ac,\m{e}{t} > \m{e}{u}, \m{e}{s} \teq \m{e}{t} - \m{e}{u}
    \end{array}
    \)}
 \end{prooftree}

 \begin{prooftree}  
  \LeftLabel{\ruleDifferenceRemove}
  \AxiomC{\( s \teq t \dsetminus u \in \Bc^*\)}
  \AxiomC{\(\m{e}{v} \in \ter{\Bc}\)}
  \AxiomC{\(v \in \{s, t, u\}\)}
  \TrinaryInfC{\(
    \begin{array}{c}
      \Ac := \Ac, \m{e}{u} \teq 0, \m{e}{s} \teq \m{e}{t}
      ~~~\vert\vert~~~
      \Ac := \Ac, \m{e}{u} \tneq 0, \m{e}{s} \teq 0
    \end{array}
    \)}
 \end{prooftree}

  \begin{prooftree}  
    \LeftLabel{\ruleSetof}
    \AxiomC{\( s \teq \setof(t) \in \Bc^*\)}
    \AxiomC{\(\m{e}{v} \in \ter{\Bc^*}\)}
    \AxiomC{\(v \in \{s, t\}\)}
    \TrinaryInfC{\(
      \begin{array}{c}
        \Ac := \Ac,1 \leq \m{e}{t}, \m{e}{s} \teq 1
        ~~~\vert\vert~~~
        \Ac := \Ac,\m{e}{t} \leq 0, \m{e}{s} \teq 0
      \end{array}
      \)}
   \end{prooftree}

 \caption{Bag rules.}
 \label{fig:bag_rules}
\end{figure}

\begin{figure}[t]
  \centering

 \begin{prooftree}  
  \LeftLabel{\ruleTableProductUp}
  \AxiomC{\(\Ac \models_{\nia} 1 \leq \m{\tup{\seq x m}}{s}
            \wedge 1 \leq \m{\tup{\seq y n}}{t}\)
   }
  \AxiomC{\(s \product t \in \ter{\Bc}\)}
  \BinaryInfC{\(\Ac := \Ac,
    \m{\tup{{\seq x m},{\seq y n}}}{s \product t} \teq
    \m{\tup{\seq x m}}{s} * \m{\tup{\seq y n}}{t}
    \)}
 \end{prooftree}

 \begin{prooftree}  
  \LeftLabel{\ruleTableProductDown}
  \AxiomC{\(\Ac \models_{\nia} 1 \leq \m{\tup{\seq x m, \seq y n}}{s \product t}\)}
  \UnaryInfC{\(\Ac := \Ac,
    \m{\tup{\seq x m, \seq y n}}{s \product t} \teq
    \m{\tup{\seq x m}}{s} * \m{\tup{\seq y n}}{t}
    \)}
 \end{prooftree}

 \begin{prooftree}  
    \AxiomC{\(\begin{array}{c}
    \Ac \models_{\nia} 1 \leq \m{\tup{\seq x m}}{s} \wedge 1 \leq \m{\tup{\seq y n}}{t}
    \\[.7ex]
    s \tabjoin_{i_1j_1 \cdots i_pj_p} t \in \ter{\Bc} 
    \quad
    x_{i_1} \teq y_{j_1},  \dots,  x_{i_p} \teq y_{j_p} \in \Bc^* 
    \end{array}\)}
    \LeftLabel{\ruleTableJoinUp}
    \UnaryInfC{\(\Ac := \Ac,
      \m{\tup{\seq x m, \seq y n}}{s \tabjoin_{i_1j_1 \cdots i_pj_p} t} \teq
      \m{\tup{\seq x m}}{s} * \m{\tup{\seq y n}}{t}
      \)}
 \end{prooftree}
 
 \begin{prooftree}  
  \LeftLabel{\ruleTableJoinDown}
  \AxiomC{\(\Ac \models_{\nia} 1 \leq \m{\tup{\seq x m, \seq y n}}
  {s \tabjoin_{i_1j_1 \cdots i_pj_p} t}\)}
  \UnaryInfC{\(
  \begin{array}{c}
    \Ac := \Ac,
    \m{\tup{\seq x m, \seq y n}}{s \tabjoin_{i_1j_1 \cdots i_pj_p} t} \teq
    \m{\tup{\seq x m}}{s} * \m{\tup{\seq y n}}{t}  \\
    \Bc := \Bc,  x_{i_1} \teq y_{j_1},  \dots,  x_{i_p} \teq y_{j_p} 
  \end{array} 
    \)}
 \end{prooftree}

  \caption{Table rules. The syntax \(\bar{x}_m\) abbreviates \(x_0, \dots, x_m\).}
  \label{fig:table_rules}
\end{figure}


\paragraph{\bf The Derivation Rules.}
The rules of our calculus are provided
in Figures~\ref{fig:bag_rules}, \ref{fig:table_rules} and \ref{fig:tables_ho_rules}.
They are expressed in \define{guarded assignment form} where
the premise describes the conditions on the current configuration
under which the rule can be applied, and
the conclusion is either \unsat, or otherwise describes
changes to the current configuration.
Rules with two conclusions, separated by the symbol \(\vert \vert\),
are non-deterministic branching rules.

In the rules, we write \(\Bc, c\), as an abbreviation of \(\Bc \cup \{c\}\) and denote by
\(\mathcal{T}(\Bc)\) the set of all terms and subterms occurring in \(\Bc\).
Premises of the form \(\Ac \models_{\nia} c \), where \(c\) is
an arithmetic constraint, can be checked by a solver
for (nonlinear) integer arithmetic.\footnote{%
  A linear arithmetic solver is enough for problems not containing the \(\product\) operator.
  For problems with SQL joins, whose encoding to SMT does require the \(\product\) operator, a solver for nonlinear arithmetic is needed, 
  at the cost of losing decidability in that case.
}
Premises of the form \(\Ec \models_\E \bot\) are checked 
by the element solver discussed earlier.

We define the following closure for \(\Bc\) where
\(\models_{\text{tup}}\) denotes entailment in the
theory of tuples, which treats all other symbols as uninterpreted functions.
%
\begin{align}
  \mathcal{\Bc}^*  &= \ 
    \{s \teq t \ver s, t \in \ter{\Bc}, \Bc \models_{\text{tup}} s\teq t\} \:\cup\:
     \{m(e, s)\teq m(e,t) \mid \Bc \models_{\text{tup}} s \teq t,\ m(e,s) \in \ter{\Bc} \} \ \nonumber \\
     &\cup \
     \{m(e_1, s)\teq m(e_2,s) \mid \Bc \models_{\text{tup}} e_1 \teq e_2,\ m(e_1,s) \in \ter{\Bc} \} \ \nonumber \\
   {} &\cup \ \{s \tneq t \ver s, t \in \ter{\Bc},
  \Bc \models_{\text{tup}} s \teq s' \wedge t \teq t' \text{ for some } s' \tneq t' \in \Bc\} \label{eq:b_star}
\end{align}

\noindent
The set \(\Bc^*\) is computable by extending standard congruence closure procedures 
with rules for adding equalities of the form
\(\select_i(\tup{x_0, \ldots, x_i, \ldots, x_k}) \teq x_i\)
and rules for deducing consequences of equalities of the form
\(\tup{s_1, \dots, s_n} \teq  \tup{t_1, \dots, t_n}\).

%
%
%

Among the derivation rules in Figure~\ref{fig:bag_rules},
%
rules \ruleAConf, \ruleEConf are applied when conflicts
are found by the arithmetic solver or the element solver.
Likewise, rule \ruleBConf is applied when the congruence closure procedure
finds a conflict between an equality and a disequality constraint.
Rules \ruleBAProp, \ruleBEProp, and \ruleAProp communicate equalities between
the three solvers.
Rule \ruleBagDisequality handles disequality between two bags \(s,t\)
by stating that some element, represented by a fresh variable \(w\),
occurs with different multiplicities in \(s\) and \(t\).
Rule \ruleNonNegative ensures that all multiplicities are nonnegative.
Rule \ruleBagEmpty enforces zero multiplicity for elements to a bag
that is provably equal to the empty bag.

Rules \ruleBagConstructorOne and \ruleBagConstructorTwo reason about
singleton bags, denoted by terms of the form \(\bag(e,n)\).
The first one splits on whether \(n\) is positive or not to determine
whether \(\bag(e,n)\) is empty, and if not, it also determines the multiplicity
of element \(e\) to be \(n\).
The second one ensures that no elements different from \(e\) are
in \(\bag(e,n)\).
%
Rules \ruleDisjointUnion, \ruleMaxUnion, \ruleBagIntersection, \ruleDifferenceSubtract,
\ruleDifferenceRemove, and \ruleSetof correspond directly to the semantics of their operators.
For example, the \ruleDisjointUnion rule applies to any multiplicity term related to bags
\((t, u, t \squplus u)\) or their equivalence classes if \(t \squplus u \in \ter{B}\).

The rules in Figure~\ref{fig:table_rules} are specific to table operations.
\ruleTableProductUp and \ruleTableProductDown are upward and downward rules 
for the \(\product\) operator.
They are the ones which introduce nonlinear arithmetic constraints on multiplicities.
\ruleTableJoinUp and \ruleJoinDown are similar to the product rules.
However, they consider the equality constraints between joining columns, to account
for the semantics of inner joins.
%


\begin{figure}
\centering
 \begin{prooftree}
    \LeftLabel{\ruleBagFilterUp}
    \AxiomC{\( \Ac \models_{\nia} 1 \leq \m{e}{t}\)}
    \AxiomC{\( \m{e}{t} \in \ter{\Bc^*}\)}
    \AxiomC{\(s \teq \filter(p, t) \in \Bc^*\)}
    \TrinaryInfC{
      \(
        \Ec := \Ec, p(e) \quad \Ac := \Ac , \m{e}{s} \teq \m{e}{t}
        ~~\vert\vert~~
        \Ec := \Ec, \neg p(e) \quad \Ac := \Ac, \m{e}{s} \teq 0
      \)
    }
 \end{prooftree}

 \begin{prooftree}
    \LeftLabel{\ruleBagFilterDown}
    \AxiomC{\( \Ac \models_{\nia} 1 \leq \m{e}{s}\)}
    \AxiomC{\( \m{e}{s} \in \ter{\Bc^*}\)}
    \AxiomC{\( s \teq \filter(p, t) \in \Bc^*\)}
    \TrinaryInfC{\(\Ec := \Ec, p(e) \qquad \Ac := \Ac , \m{e}{s} \teq \m{e}{t}\)
    }
  \end{prooftree}

  \begin{prooftree}
    \LeftLabel{\ruleBagMapUp}
    \AxiomC{\( \Ac \models_{\nia} 1 \leq \m{e}{t}\)}
    \AxiomC{\( \m{e}{t} \in \ter{\Bc^*}\)}
    \AxiomC{\(s \teq \map(f,t) \in \Bc^*\)}
    \AxiomC{\(e \not\in W\)}
    \QuaternaryInfC{\(\Ac := \Ac , \m{e}{t} \leq \m{f(e)}{s}\) }
 \end{prooftree}

 \begin{prooftree}
    \LeftLabel{\ruleBagMapDownInjective}
    \AxiomC{\(\m{e}{s} \in \ter{\Bc^*}\)}
    \AxiomC{\(s \teq \map(f,t) \in \Bc^*\)}
    \AxiomC{\(f\) is injective}       
    \TrinaryInfC{\(\Ec := \Ec, f(w) \teq e \qquad \Ac:=\Ac,\m{e}{s} \teq \m{w}{t}\)}
 \end{prooftree}

 \begin{prooftree}
    \LeftLabel{\ruleBagMapUpNonInjectiveUp}
    \AxiomC{\(\Ac \models_{\nia} 1\leq \m{e}{t}\)}
    \AxiomC{\(s \teq \map(f,t) \in \Bc^*\)}
    \BinaryInfC{\(\Bc := \Bc, i \teq \elementIndex(e, t) \qquad \Ac := \Ac , 1 \leq i \leq \delem(t) \) }
 \end{prooftree}

 \begin{prooftree}
    \LeftLabel{\ruleBagMapUpNonInjectiveDown}
    \AxiomC{\(\Ac \models_{\nia} 1\leq \m{e}{s} \quad
            s \teq \map(f,t)\in \Bc^*\)}
    \UnaryInfC{\( \Ac := \Ac , \mapSum(e, t, \delem(t)) \teq \m{e}{s}, \mapSum(e, t, 0) \teq 0 \) }
 \end{prooftree} 
  \caption{Bag filter and map rules. \(w,i\) are fresh variables.}
  \label{fig:tables_ho_rules}
\end{figure}

The rules in Figure~\ref{fig:tables_ho_rules} reason about
the filter (\(\filter\)) and map (\(\map\)) operators.
\ruleBagFilterUp splits on whether an element \(e\) in \(s\) satisfies
(the predicate denoted by) \(p\) or not in order to determine its multiplicity 
in bag \(\filter(p, s)\).
%
\ruleBagFilterDown concludes that every element with positive multiplicity
in \(\filter(p, s)\) necessarily satisfies \(p\) and 
has the same multiplicity in \(s\).
\ruleBagMapUp applies the function symbol \(f\) to every element \(e\)
that is provably in bag \(t\).
Note that it cannot determine the exact multiplicity of \(f(e)\)
in bag \(\map(f, t)\)
since multiple elements can be mapped to the same one by \(f\)
if (the function denoted by) \(f\) is not injective.
Therefore, the rule just asserts that \(\m{f(e)}{\map(f, t)}\) is
at least \(\m{e}{t}\).
To prevent derivation cycles with \ruleBagMapDownInjective,
rule \ruleBagMapUp applies only if \(e\) is not a variable introduced
by the downward rule.

The downward direction for map terms is more complex, and expensive,
if \(f\) is not injective.
Therefore, before solving, we check the injectivity of each function symbol \(f\)
occurring in map terms.
This is done via a subsolver instance that checks the satisfiability
of the formula \(f(x) \teq f(y) \wedge x \tneq y\) for fresh variables \(x, y\).
If the subsolver returns \unsat, which means that \(f\) is injective,
%
we apply rule \ruleBagMapDownInjective
which introduces a fresh variable \(w\) for each term \(\m{e}{\map(f,t)}\)
such that \(e = f(w)\) and \(w\) is in \(t\) with the same multiplicity.
In contrast, if the subsolver returns \sat or \unknown, 
we treat \(f\) as non-injective and rely on a number of features of \cvc
to construct and process a set of quantified constraints which, %
informally speaking and mixing syntax and semantics here for simplicity,
formalize the following relationship between the multiplicity of an element \(e\) 
in a bag \(\map(f,t)\) and that of the elements of \(t\) that \(f\) maps to \(e\):
\begin{align}
  \m{e}{\map(f,t)} = \sum\set{\m{x}{t} \mid x \in t \,\wedge\, f(x) = e} \label{eq:map_multiplicity}
\end{align}
To encode this constraint, we introduce three uninterpreted symbols, 
\(\delem: \ms(\alpha) \rightarrow \int\), 
\(\elementIndex: \alpha \times \ms(\alpha) \rightarrow \int \), and 
\(\mapSum: \alpha \times \ms(\alpha) \times \int \rightarrow \int\).
The value \(\delem(t)\) represents the number of distinct elements 
in (the bag denoted by) \(t\);
\(\elementIndex(e, t)\) represents a unique index in the range \([1,\delem(t)]\)
for element \(e\) in bag \(t\);
for an element \(e\) in \(\map(f,t)\), 
\(\mapSum(e, t, i)\) accumulates the multiplicities of the elements of \(t\) 
with index in \([1, i]\) that \(f\) maps to \(e\).
Rule \ruleBagMapUpNonInjectiveUp ensures that every element in \(t\) is assigned
an index \(i\) in \([1, \delem(t)]\), 
whereas \ruleBagMapUpNonInjectiveUp constrains the multiplicity \(\m{e}{s}\) 
to be \(\mapSum(e, t, \delem(t))\) when \(e\) is in \(s\).  
\begin{paper}%
We do not describe the encoding of (\ref{eq:map_multiplicity}) here 
due to space limitations.
However, it is an axiom with bounded quantification over the interval 
\([1,\delem(t)]\)
that is processed by \cvc's model-based quantifier instantiation 
module~\cite{boundedQuantifiers}.
\end{paper}%
\begin{report}%
The encoding of the quantified formula is given in~\ref{eq:map_reduction}.
\end{report}

%

%
%
%

%
\begin{example}\label{ex:unsat_sat_map}
  %
  Suppose we have the constraints:
  \(\set{x \not\sqin s, y \sqin \map(f, s), y \teq x + 1}\)
  where \(x,y \) are integers and \(f\) is defined in the arithmetic solver to be
  the integer successor function (\(f(x) = x + 1\)). 
  After applying simplification rules in Figure~\ref{fig:rewriting},
  we end up with \(c_0 = \conf{\Ac_0, \Bc_0, \Ec_0}\), 
  where
  \(\Ac_0 = \set{\neg(1 \leq \m{x}{s}), 1 \leq \m{y}{\map(f, s)}, y \teq x + 1}\)
  and \(\Ec_0 = \set{ y \teq x + 1 }\).
  Applying rule \ruleNonNegative, we get \(c_1 = \conf{\Ac_1,\Bc_1, \Ec_1}\),
  where
  \(\Ac_1 =  \Ac_0 \cup \set{0 \leq \m{x}{s},  0 \leq \m{y}{\map(f, s)}}\)
  and \(\Ec_1 = \Ec_0\).
  Since \(f\) is injective, we can apply rule \ruleBagMapDownInjective
  to get \(c_2 = \conf{\Ac_2, \Bc_2, \Ec_2}\), where
  \(\Ac_2  = \Ac_1 \cup \set{\m{w}{s}  \teq \m{y}{\map(f, s)}}\)
  and \(\Ec_2 = \Ec_1 \cup \set{y  \teq w + 1}\).
  Next, we apply the propagation rule \ruleEProp followed by 
  \ruleBAProp to get \(\Ac_3 = \Ac_2\), 
  \(\Ac_4 = \Ac_3 \cup \set{y \teq w  + 1}\).
  Now, \(\Ac_4 \) is unsatisfiable because it entails 
  \(x \teq w\), \(\neg(1 \leq \m{x}{s})\), and \(1 \leq \m{w}{s}\).
  Hence the rule \ruleAConf applies, and we get \unsat.
  If we remove the constraint \(x \not\sqin s\), then 
  the problem is satisfiable
  and we can construct a model \(\I\) where
  \(\I(x)\) is any natural \(n\), \(\I(y) = n + 1\), and 
  \(\I(s)\) is a singleton bag containing \(n\) with multiplicity 1.
\end{example}
%

%

\subsection {Calculus Correctness} \label{sec:tables_correctness}
Logozzo et al.~\cite{bagsPaper2008} proved the decidability of
the theory of bags with linear constraints over bag cardinality and the operators:
$\bempty{}$, $\bag$, $\msym$, $\sqcup$, $\squplus$,
$\sqcap$, $\setminus$, $\dsetminus$, $\setof$, $\sqin$, $\sqsubseteq$.
In this work, we exclude the cardinality operator. 
However, we prove that adding the \(\filter\) operator with computable predicates
preserves decidability (See Proposition~\ref{prop:termination_no_map} below).
In contrast, the further addition of the $\product$ and  $\project$ operators makes 
the problem undecidable.
This is provable with a reduction from the undecidable equivalence problem 
for unions of conjunctive SQL queries~\cite{unionOfConjuctiveQueries}. 
%
%
%
Our calculus is not refutation complete in general in the presence of maps
because of the possibility of nontermination, as shown in the example below.
\begin{example}\label{ex:incomplete_map_tables}
  Suppose \(x\) is an integer variable, and \(s\) is an integer bag variable
  and consider just the constraint set
  \(\{\m{x}{s} \teq 1,\, s \teq \map(f, s)\}\) where \(f\) is again the successor function.
  The set is unsatisfiable in the theory of finite bags.
  However, this is not provable in our calculus as it allows 
  the repeated application of rules \ruleBagMapUp and \ruleAProp, 
  which add fresh elements \(x+1, x+2, x+3, \dots\) to \(s\).
\end{example}

Another source of refutation incompleteness is the presence of constraints
with the cross product operator \(\product\), which causes 
the generation of nonlinear constraints for the arithmetic subsolver,
making the entailment checks in rules such as \ruleAConf and \ruleBAProp
undecidable.

However, in the absence of \(\map\) and \(\product\), 
the calculus is both refutation and solution sound, as well as terminating.
The soundness properties are a consequence of the fact that
each derivation rule preserves models, as specified in the following lemma.

\begin{restatable}{lemma}{tablesEqSatisfiable}
  \label{lem:tables_eq_satisfiable}
  For all applications of a rule from the calculus, the premise configuration  
  is satisfiable in \(\tablesTheory\) if and only if one of the conclusion configurations 
  is satisfiable in \(\tablesTheory\).
\end{restatable}

The proof of this lemma provides actually a stronger result than stated
in the right-to-left implication above:
for each of the conclusion configurations \(C'\), every model of \(\tablesTheory\) 
that satisfies \(C'\) satisfies the premise configuration as well.
This implies that any model that satisfies a saturated leaf of a derivation tree 
satisfies the root configuration as well.

\begin{restatable}[Refutation Soundness]{proposition}{refutation}
  \label{prop:refutation}
  For every closed derivation tree with root node \(C\), configuration \(C\) is
  \(\tablesTheory\)-unsatisfiable.
\end{restatable}
\noindent
The proposition above implies that deriving a closed tree with a root
\(\tup{\Ac, \Bc, \Ec}\) is sufficient to prove the unsatisfiability in \(\tablesTheory\)
of the constraints \(\Ac \cup \Bc \cup \Ec \).
The proof of the proposition is a routine proof by induction on the structure 
of the derivation tree.

Solution soundness has a more interesting proof since it relies on the construction
of a satisfying interpretation for a saturated leaf of a derivation tree,
which is a witness to the satisfiability of the root configuration.
We describe next at a high level how to construct an interpretation \(\I\) 
for a saturated configuration \(C=\conf{\Ac, \Bc, \Ec}\).\footnote{%
The full model construction and its correctness are discussed in detail in the%
\begin{paper}
longer version of this paper~\cite{arxiv}.
\end{paper}
\begin{report}
appendix.
\end{report}
}
More precisely, we construct a \emph{valuation} \(\I\) of the variables and terms in \(C\)
that agrees with the semantics of the theory symbols and satisfies 
\(\Ac \cup \Bc \cup \Ec\).

Once again, to simplify the exposition, 
we assume, \emph{with} loss of generality, that any element sort \eleSort in the problem
can be interpreted as an infinite set.
The actual implementation uses \cvc's theory combination mechanism to allow
also sorts denoting finite sets, under mild restrictions on the theories involved~\cite{politeness}.
\medskip

\noindent\textbf{Model Construction Steps}
%
\begin{enumerate}\setlength{\itemsep}{.2ex}
\item Sorts: 
The meaning of the sort constructors
\(\ms\), \(\tuple\), and \(\int\) is fixed by the theory.
The element sort \(\eleSort\) is interpreted as an infinite set.
As a concrete representation of bags, and hence of tables, consistent with the theory,
we choose sets of pairs $(e,n)$, where $e$ is an element of the bag in question, and 
$n$ is its (positive) multiplicity.
\item \(\tablesSig\): 
\(\tablesTheory\) enforces the interpretation of all \(\tablesSig\)-symbols.
Saturation guarantees that equivalent bag/tuple terms will be interpreted 
by \(\I\) as the same bag/tuple.
\item Integer variables: 
Saturation guarantees that there is some model of \(\nth\) that satisfies \(\Ac\).
We define \(\I\) to interpret integer variables according to this model.
\item Variables of sort \eleSort:
The calculus effectively partitions them into equivalent classes
(where $x$ is in the same class as $y$ iff $x \teq y$ is entailed 
by \(\Bc\)).
Each class is assigned a distinct element from \(\I(\eleSort)\), 
which is possible since it is infinite.
\item Bag variables: 
$\I$ interprets each bag variable \(s\) as the set
\[ \I(s) = \set{(\I(e), n) \mid \m{e}{s} \in \ter{\Bc^*},\; \I(\m{e}{s}) = n > 0}.\]
A well-foundedness argument on $\I$'s construction guarantees that 
the equation above is well defined.
Note that \(\I(s)\) is the empty bag iff there is no term $\m{e}{s}$ in \(\ter{\Bc^*}\) 
that satisfies the conditions in the comprehension above.
\end{enumerate}
\begin{example}\label{ex:model_construction}
  From Example~\ref{ex:unsat_sat_map}, the set of constraints \(\set{y \sqin \map(f, s), y \teq x + 1}\)
  is satisfiable. 
  In this example the element sort is \(\sint\), interpreted as the integers. 
  Similar to Example~\ref*{ex:unsat_sat_map}, the terms 
  \(\m{y}{\map(f, s)},w, \m{w}{s}\) are generated during solving. 
  After saturation, and for simplicity, suppose we end up with the following sets of 
  equivalence classes for the terms involved:
  \(\{\m{y}{\map(f, s)}\), \(\m{x}{s}\}\), \(\{x, w\}\), \(\{y, x + 1\}\), \(\{s\}\), \(\{\map(f, s)\}\). 
  The theory of arithmetic assigns consistent concrete values 
  to the first three equivalence classes, say 1, 10, 11, respectively.
  In Step 5, \(s\) is interpreted as the set 
  \(\{(\I(x), \I(\m{x}{s})),\ (\I(w), \I(\m{w}{s}))\} = \{(10, 1)\}\), 
  and similarly  \(\I(\map(f, s)) = \{(\I(y), \I(\m{y}{\map(f, s)}))\} = \{(11,1)\}\). 
\end{example}
%

\begin{restatable}[Solution Soundness]{proposition}{solution}
  \label{prop:solution}
  For every derivation tree with root node \(C_0\) and a saturated leaf \(C\), 
  configuration \(C_0\) is
  satisfiable in \(\tablesTheory\).
\end{restatable}
%

%
%

While the calculus is not terminating in general, we can show that 
all derivations are finite in restricted cases.

\begin{restatable}[Termination]{proposition}{terminationNoMap}
  \label{prop:termination_no_map}
  Let \(C\) be a configuration containing no product, join, or map terms.
  All derivation trees with root \(C\) are finite.
\end{restatable}

Under the restrictions in the proposition above,
the rules of the calculus will never generate nonlinear arithmetic constraints.
This means that for input problems that also have no such constraints,
any derivation strategy for the calculus yields a decision procedure.

%

\section{A Theory of Finite Relations with Filter and Map}\label{sec:relations_theory}
%

\begin{figure}[!tbp]
  %
  \begin{prooftree}
    \LeftLabel{\ruleRelationJoinUp}
    \AxiomC{\(\begin{array}{c}
     \tup{x_1, \dots, x_m} \sqin s \in \Sc^*
      \qquad 
      \tup{y_1, \dots, y_n} \sqin t \in \Sc^*
      \\[.7ex]
      s \tabjoin_{i_1j_1 \cdots i_pj_p} t \in \ter{\Sc}
      \qquad
      x_{i_1} \teq y_{j_1},  \dots,  x_{i_p} \teq y_{j_p} \in \Sc^* 
     \end{array}\)}
    \UnaryInfC{\(\Sc := \Sc,
      \tup{x_1, \dots, x_m,y_1, \dots, y_n} \sqin s \tabjoin_{i_1j_1 \cdots i_pj_p} t
      \)}
  \end{prooftree}

  \begin{prooftree}
    \LeftLabel{\ruleRelationJoinDown}
    \AxiomC{\(\tup{x_1, \dots, x_m,y_1, \dots, y_n} \sqin
    {s \tabjoin_{i_1j_1 \cdots i_pj_p} t} \in \Sc^*\)}
    \UnaryInfC{\(
      \begin{array}{c}
        \Sc := \Sc,        
        \tup{x_1, \dots, x_m} \sqin s, \tup{y_1, \dots, y_n} \sqin t,
         x_{i_1} \teq y_{j_1},  \dots,  x_{i_p} \teq y_{j_p}
      \end{array}
      \)}
  \end{prooftree}

  \begin{prooftree}
    \LeftLabel{\rn{Filter up}}
    \AxiomC{\(e \sqin s \in \Sc^*\)}
    \AxiomC{\(\filter(p, s) \in \ter{\Sc}\)}
    \BinaryInfC{
      \(\Ec:= \Ec, p(e) \quad \Sc := \Sc, e \sqin \filter(p, s)
       ~~~\vert\vert~~~
       \Ec:= \Ec, \neg p(e)\quad \Sc := \Sc, e \not\sqin \filter(p, s)\)
    }
  \end{prooftree}

  \begin{prooftree}
    \LeftLabel{\rn{Map up}}
    \AxiomC{\( e \sqin s \in \Sc^*\)}
    \AxiomC{\(\map(f,s) \in \ter{S}\)}
    \AxiomC{\(e \not\in W\)}
    \TrinaryInfC{\(\Sc := \Sc , f(e) \sqin \map(f,s)\) }
  \end{prooftree}

  \begin{prooftree}
    \LeftLabel{\rn{Map down}}
    \AxiomC{\(e \sqin \map(f,s) \in \Sc^*\)}
    \AxiomC{\(w \in W\)}
    \BinaryInfC{\(\Ec := \Ec, f(w) \teq e \quad \Sc:=\Sc, w \sqin s\)}
    \DisplayProof \hskip 1em
    \AxiomC{\( e \sqin \filter(p, s) \in \Sc^*\)}
    \LeftLabel{\rn{Filter down}}
    \UnaryInfC{\(\Ec := \Ec, p(e) \quad \Sc := \Sc, e \sqin s\)
    }
  \end{prooftree}

  \caption{Set filter and map rules. \(w\) is a fresh variable unique per \(e, f,s\).}
  \label{fig:relations_ho_rules}
\end{figure}

To reason about SQL queries under set semantics,
we rely on the theory of finite relations \(\relationsTheory\),
whose signature and calculus are described in Meng et al.~\cite{relationsPaper}.
We extend the signature with filter and map operators analogous 
to those in the theory of tables.
%
We overload the symbols 
\(\sqin,\sqsubseteq, \sqcup, \sqcap, \setminus, \product, \filter, \map,  
\tabjoin\)
for set operations membership, subset, union, intersection, difference, product,
filter, map, and (inner) join,\footnote{%
This is analogous to the inner join operator for tables
but differs from the relational join operator defined in Meng et al.~\cite{relationsPaper}.
} 
respectively. 
We use \(\relationProject\) for relation projection. 

A calculus to reason about constraints in \(\relationsTheory\) can be defined 
over configurations of the form \(\tup{\Sc, \Ec}\), where
\(\Sc\) is a set of \(\relationsTheory\)-constraints,
and \(\Ec\) is a set of element constraints similar to the one defined
in Section~\ref{sec:tables_theory}.
We again assume we have a solver that can decide the satisfiability 
of element constraints.
The \(\relationsTheory\)-constraints are
equational constraints of the form \(s \teq t\) and \(s \tneq t\) and
membership constraints of the form \(e \sqin s\) and \(e \not\sqin s\),
where \(e, s, t\) are \(\Sigma_{r}\)-terms.
Similar to bags, we also define the set \(W\) to be an infinite set of fresh variables
specific for map rules.
%
%
The closure \(\Sc^*\) for the \(\Sc\) component of a configuration is defined as:
\begin{align*}
  \mathcal{\Sc}^* = \ 
   & \{s \teq t \ver s, t \in \ter{\Sc}, \Sc \models_{\text{tup}} s\teq t\} \\
   \cup \ 
   & \{s \tneq t \ver s, t \in \ter{\Sc},
       \Sc \models_{\text{tup}} s \teq s' \wedge t \teq t' \text{ for some } s' \tneq t' \in \Sc\} \\
   \cup \ 
   & \{e \sqin s \ver e, s \in \ter{\Sc},
       \Sc \models_{\text{tup}} e \teq e' \wedge s \teq s' \text{ for some } e' \sqin s' \in \Sc\} 
\end{align*}

\begin{paper}%
For space constraints, we refer the reader to the extended version
of this paper~\cite{arxiv} for a full description of the extended signature and calculus.
\end{paper}%
\begin{report}%
The extended signature and the original derivation rules are provided 
in Appendix~\ref{sec:appendix_sets} for reference.
\end{report}
We provide in Figure~\ref{fig:relations_ho_rules} only our additional rules:
those for filter and map, 
which resemble the corresponding rules given in Figure~\ref{fig:tables_ho_rules} 
for bags, 
and those for inner joins.
%
%

\paragraph{Calculus correctness.}
Bansal et al.~\cite{DBLP:journals/lmcs/BansalBRT18} provide
a sound, complete, and terminating calculus 
for a theory of finite sets with cardinality.
The calculus is extended to relations but without cardinality
by Meng et al.~\cite{relationsPaper}.
That extension is refutation- and solution-sound but not terminating in general.
While it is proven terminating over a fragment of the theory of relations,
that fragment excludes the relational product operator,
which we are interested in here.

We have proved that the calculus by Meng et al.~extended with the filter rules 
in Figure~\ref{fig:relations_ho_rules} is terminating for constraints built 
with the operators \(\{\teq, \sqcup, \sqcap, \setminus, \product, \filter\}\)%
\begin{paper}%
~\cite{arxiv}.%
\end{paper}
\begin{report}%
(see Lemma~\ref{newlemma}).
\end{report}
Termination is lost with the addition of map rules,
also in Figure~\ref{fig:relations_ho_rules}.
Example~\ref{ex:incomplete_map_tables} works in this case too as a witness.

However, termination can be recovered when the initial configuration satisfies 
a certain \emph{acyclicity} condition.
To express this condition, we associate with each configuration \(C = \tup{\Sc, \Ec}\) 
an undirected multi-graph \(G\) 
whose vertices are the relation terms of sort \(\s\) occurring in $\Sc$, and 
whose edges are labeled with an operator in \(\{\teq, \sqcup, \sqcap, \setminus, \product, \tabjoin, \filter, \map\}\).
Two vertices have an edge in \(G\) if and only
if they are in the same equivalence class or a membership rule can be applied 
(either upward or downward) to \(C\) such that one of the vertices occurs 
as a child of the other vertex in the premises or conclusions of the rule.   
Furthermore, each edge between vertices \(\map(f, s)\) and \(s\) is colored red
while the remaining edges are colored black.

\begin{restatable}{proposition}{mapTermination}
  \label{thm:map_termination}
  Let \(C_0 = \tup{\Sc_0, \Ec_0}\) be an initial configuration and 
  let \(G_0\) be its associated multi-graph.
  The calculus is terminating for \(C_0\)
  %
  if there is no cycle in \(G_0\) with red edges, 
  and all cycles are located in a subgraph without map terms. 
\end{restatable}
Figure~\ref{fig:redCycle} shows the graph for Example~\ref{ex:incomplete_map_tables}
which has a cycle with a red edge.
Figure~\ref{fig:no_red_cycle} shows a graph with no red cycles. 
For the purposes of this paper, the acyclicity condition is not a serious restriction
because typical SQL queries translate to initial configurations
whose associated graph does not have cycles with red edges.\footnote{%
Examples of SQL queries that do have cycles with red edges are queries 
with recursive common table expressions.
}

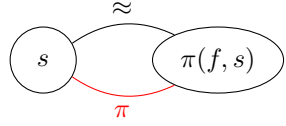
\begin{figure}[t]
  \centering
  \begin{tikzpicture}[auto]
    \node[state, ellipse] (a)                  {\(s\)};
    \node[state, ellipse] (mapA)[right=of a]{\(\map(f, s)\)};
    \path[](a) edge [bend left]  node  [above] {\(\teq\)} (mapA);
    \path[red](a) edge [bend right]  node  [below] {\(\map\)} (mapA);
  \end{tikzpicture}
  \caption{A cycle that contains a red edge with \(\map\) label. }
  \label{fig:redCycle}
\end{figure}
\begin{figure}[t]
  \centering
  \begin{adjustbox}{width=.7\textwidth,center}
  \begin{tikzpicture}[auto]
    \node[state, ellipse] (a)  at (0,0)                {\(s_i\)};
    \node[state, ellipse] (b) at (-4, 0) {\(\map(f_i, s_i)\) };
    \node[state, ellipse] (x)[above left=of b,yshift=-1cm]{\(\cdots\)};
    \node[state, ellipse] (y)[below left=of b,yshift=0.5cm]{\(\cdots\)};
    \node[state, ellipse] (d)[above right=of a,yshift=-1cm]{\(t\)};
    \node[state, ellipse] (e)[below right=of a,yshift=0.5cm]{\(s_j\)};
    \node[state, ellipse] (c) [right=of e,xshift=1.3cm]  {\(\map(f_j, s_j)\) };
    \path[red](a) edge node  [below] {\(\map\)} (b);
    \path[red](c) edge node  [below] {\(\map\)} (e);
    \path[](b) edge node  [above] {\(\teq\)} (x);
    \path[](b) edge node  [below] {\(\sqcup\)} (y);
    \path[](a) edge node  [above] {\(\teq\)} (d);
    \path[](a) edge node  [left] {\(\sqcap\)} (e);
    \path[](d) edge node  [right] {\(\setminus\)} (e);
    \draw[] (-0.6,-1.8) -- (-0.6, 1.3) -- (2.5,1.3) -- (2.5,-1.8) -- (-0.6, -1.8);
  \end{tikzpicture}
\end{adjustbox}
  \caption{A graph with no cycles that contain red edges.}
  \label{fig:no_red_cycle}
\end{figure}
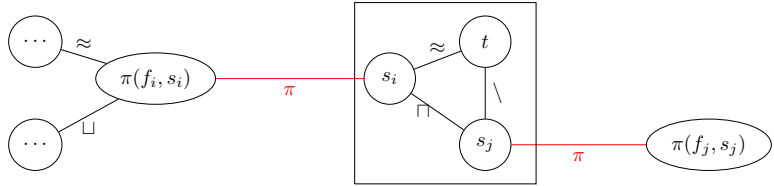
%

%

\section{Supporting Nullable Sorts}\label{sec:nullables}
SQL and relational databases allow tables with \emph{null values}.
To support them in SMT, one could rely in principle on SMT solvers 
that accept user-defined algebraic datatypes~\cite{datatypesPaper},
as nullable types are a form of option types.
That is, however, not enough since, to encode SQL operations on nullable types,
one also needs to lift all operators over a given sort to its nullable version. 
Since this is extremely tedious and error-prone when done at the user level,
we extended the datatypes solver in cvc5 by adding built-in
parametric \(\nullable\) sorts, along with selectors, testers, 
and a family of (second-order) lifting operators.
The signature of the corresponding theory of \define{nullable sorts}
is provided in Figure~\ref{fig:nullables_sig}.
The meaning of constructors \(\nullN\) and \(\someN\),
the selector \(\valN\), and the testers \(\isNull, \isSome\) is standard;
that of \(\lift\) is:
\begin{align*}
  \lift(f, x_1, \dots, x_k) =
  \begin{cases}
    \nullN                                  & \text{if } x_i = \nullN \text{ for some } i \in [1, k]     \\
    \someN(f(\valN(x_1), \dots, \valN(x_k))) & \text{otherwise}
  \end{cases}
\end{align*}
which is analogous to the semantics of eager evaluation in programming languages.
This semantics covers most SQL operations, except for \sql{OR} and \sql{AND}, 
as SQL adopts a three-valued logic for nullable Booleans~\cite{theCompleteBook},
interpreting, for instance, \sql{NULL OR TRUE} as \sql{TRUE} and
\sql{NULL AND FALSE} as \sql{FALSE}.  
We can support this short-circuiting semantics through an encoding from SQL to SMT
based on the \textit{if-then-else} operator.
%

\begin{figure}[t]
  \scriptsize\centering
  \(
    \begin{array}{l}
      \begin{array}{@{}l@{\quad}l@{\quad}l@{\quad}l@{}}
        \toprule
        \textbf{Symbol} & \textbf{Arity}                     & \textbf{SMTLIB}         & \textbf{Description}    \\
        \midrule
        \nullN          & \nullable(\alpha)                   & \verb|nullable.null|    & \text{Null constructor} \\
        \someN          & \alpha \rightarrow \nullable(\alpha) & \verb|nullable.some|    & \text{Some constructor} \\
        \valN           & \nullable(\alpha) \rightarrow \alpha & \verb|nullable.val|     & \text{Value selector}   \\
        \isNull         & \nullable(\alpha) \rightarrow \bool & \verb|nullable.is_null| & \text{Is null tester}   \\
        \isSome         & \nullable(\alpha) \rightarrow \bool & \verb|nullable.is_some| & \text{Is some tester}   \\
        \lift           & \tiny
        \left( \alpha_1 \times \dots \times \alpha_k \rightarrow \alpha \right)
        \times         \\
                        & \tiny
        \nullable(\alpha_1) \times \dots \times \nullable(\alpha_k)
        \rightarrow \nullable(\alpha)
                        & \verb|nullable.lift|
                        & \text{Lifting operators}
        \\
        \bottomrule
      \end{array}
      \\[4ex]
    \end{array}
  \)
  \caption{Signature for the theory of nullable sorts. \(\lift\) is a variadic function symbol.}
  \label{fig:nullables_sig}
\end{figure}


%

\section{Evaluation of Benchmarks on Sets and Bags}\label{sec:evaluation}

We evaluated our cvc5 implementation using a subset of the benchmarks~\cite{calcite-benchmarks} 
derived from optimization rewrites of  Apache \calcite, 
an open source database management framework~\cite{calcite}.
The benchmark set provides a database schema for all the benchmarks.
Each benchmark contains two queries over that schema 
which are intended to be equivalent under bag semantics, 
in the sense that they should result in the same table for any instance 
of the schema.
The total number of available benchmarks is 232. 
We modified many queries that had syntax errors or 
could not actually be parsed by \calcite.
We then excluded benchmarks with queries containing constructs 
we currently do not support, 
such as \sql{ORDER BY} clauses\footnote{%
Both \spes and \sqlSolver provide partial support for \sql{ORDER BY}.
They classify two input queries with \sql{ORDER BY} clauses as equivalent 
if they can prove their respective subqueries 
without the \sql{ORDER BY} clause equivalent.
Otherwise, they return \unknown.
}
or aggregate functions, or benchmarks
converted to queries with aggregate functions by the \calcite parser.
%
%
That left us with 88 usable benchmarks, that is, about 38\% of all benchmarks.
%
%
%
%
%
For the purposes of the evaluation we developed a prototype translator 
from those benchmarks into SMT problems over the theory of tables and of relations
presented earlier, following SQL's bag semantics and set semantics, respectively.
The translator uses different encodings for bag semantics and set semantics. 
Each SMT problem is unsatisfiable iff the SQL queries in the corresponding benchmark
are equivalent under the corresponding semantics.

We ran \cvc on the translated benchmarks on a computer with 128GB RAM and with a
12th Gen Intel(R) Core(TM) i9-12950HX processor.  We
compared our results with those of \sqlSolver and \spes analyzers ---
which can both process the calcite benchmarks directly. 
The results are shown in Figure~\ref{fig:eq_benchmarks}.
The first three lines show the results for the bags semantics encoding (b) 
while the fourth line shows the results for the set semantics encoding (s).
The column headers  \(\not\equiv\), \(\equiv\), and uk stand 
for inequivalent, equivalent, and \unknown respectively.
Unknown for \spes means it returns ``not proven,'' 
whereas for \cvc it means that it timed out after 10 seconds. 
\sqlSolver solves all benchmarks efficiently within few seconds.
However, it incorrectly claims that the two queries in one benchmark,
\texttt{testPullNull},
are equivalent, despite the fact that they return tables which differ
by the order of their columns.
The \sqlSolver developers acknowledged this issue as an error 
after we reported it to them, and they fixed it in a later version. 
%
%
\spes too misclassified benchmark \texttt{testPullNull}
as well as another one (\texttt{testAddRedundantSemiJoinRule}) 
where the two queries are equivalent only under set semantics.
The \spes authors acknowledged this as an error in their code. 
The tool is supposed to answer \unknown for the second benchmark, 
as the two queries have different structures, 
a case that \spes does not support.
\cvc gives the correct answer for these two benchmarks.
In each case, it provides counterexamples in the form of a small database 
over which the two SQL queries differ, something that we were able to verify
independently using the \postgres database server~\cite{postgresql}.
Under set semantics, \cvc 
found 83 out of 88 benchmarks to contain equivalent queries, 
found 1 to contain inequivalent queries (\texttt{testPullNull}), and 
timed out on 4 benchmarks.
%
%
Under bag semantics, \cvc proved fewer benchmarks than both \spes and \sqlSolver.
However, it found the benchmark whose queries are inequivalent 
under bag semantics but equivalent under set semantics.

Since the \calcite benchmark set is heavily skewed towards equivalent queries,
we mutated each of the 88 benchmarks to make the mutated queries inequivalent.
The mutation was performed manually but \emph{blindly},
to avoid any bias towards the tools being compared.
For the same reason, we excluded the two original benchmarks 
that \spes misclassifies.
The results of running the three tools on the mutated benchmarks are shown 
in Figure~\ref{fig:neq_benchmarks}.
\sqlSolver was able to solve all of them correctly.
As expected, \spes returned \unknown on all of them 
since it cannot verify that two queries are inequivalent.

Note that \cvc's performance improves with the mutated benchmarks 
because many of them use non-injective map functions, and 
it is easier to find countermodels for such benchmarks than it is to prove equivalence. 
Consistent with that,
we observe that \cvc times out for more benchmarks under set semantics 
than for bag semantics.
We conjecture that this is because more queries are equivalent 
under set semantics than they are under bag semantics. 
%

Looking at \cvc's overall results, more benchmarks were proven equivalent by \cvc under set
semantics than bag semantics. 
%
How to improve performance under bag semantics requires further investigation.
However, we find it interesting and useful for the overall development 
of SQL analyzers that \cvc was able to expose a couple of bugs in the compared tools.

\begin{figure}[t]
\quad~
    \begin{subfigure}[b]{0.37\textwidth}
      \begin{tabular}{llrrrr}
        \toprule
                 &  & \(\not\equiv\) & \(\equiv\) & uk & Total \\
        \midrule
        \sqlSolver &(b) & 1              & 87         &         & 88    \\
        SPES     &(b)  &                & 54         & 34      & 88    \\        
        cvc5 &(b) & 2              & 42         & 44      & 88    \\
        \hline
        cvc5 &(s) & 1              & 83         & 4       & 88    \\
        \bottomrule
      \end{tabular}
      \caption{}
      \label{fig:eq_benchmarks}
    \end{subfigure}
    \qquad\qquad
    \begin{subfigure}[b]{0.4\textwidth}
      \begin{tabular}{llrrrr}
        \toprule
                   & & \(\not\equiv\) & \(\equiv\) & uk & Total \\
        \midrule
        \sqlSolver &(b)& 86             &            &         & 86    \\
        SPES       &(b)&                &            & 86      & 86    \\        
        cvc5 &(b) & 81             &            & 5       & 86    \\
        \hline
        cvc5 &(s)  & 67             & 9          & 10      & 86    \\
        \bottomrule
      \end{tabular}
      \caption{}
      \label{fig:neq_benchmarks}
    \end{subfigure}
  \caption{Left table is the original benchmarks, right table is the mutated one.
    SPES can only answer equivalent or \unknown.
  }
  \label{fig:sat_results}
\end{figure}

%

\section{Conclusion and Future Work}\label{sec:conclusion}
We showed how to reason about SQL queries automatically using a number of theories in SMT solvers.
We introduced a theory of finite tables to support SQL's bag semantics.
We extended a theory of relations 
to support SQL's set semantics.
We handled null values by extending a theory of algebraic datatypes with nullable sorts and
a generalized lifting operator.
We also showed how to translate SQL queries without aggregation into these theories.
Our comparative evaluation has shown that our implementation is not yet
fully competitive performance-wise with that of specialized SQL analyzers, 
particularly on equivalent queries with non-injective mapping functions.
We plan to address this in future work.

We are experimenting with adding support for \emph{fold} functionals
in order to encode and handle SQL queries with aggregation operators.
%
%
Another direction for future work is adding filter and map operators
to a theory of sequences~\cite{sequencesPaper} to reason about SQL queries
with \emph{order-by} clauses, which current tools support only in part.

%

\paragraph{Acknowledgements}
This work was supported in part by a gift from Amazon Web Services.
We are grateful to Qi Zhou and Joy Arulraj, the authors of \spes,
for their help, detailed answers to our questions, and 
for sharing the source code and benchmarks of \spes, which we used in
this paper.
We are also grateful to Haoran Ding, the first author of \sqlSolver,
for his prompt answers to our questions and his clarifications on some aspects 
of the tool. 
Finally, we thank Abdalrhman Mohamed 
\begin{paper}%
and the anonymous reviewers for their 
\end{paper}%
\begin{report}%
for his
\end{report}%
feedback and suggestions for improving the paper.

\bibliographystyle{splncs04}
\bibliography{biblio.bib}
\begin{report}
\appendix
\section{Proofs for Tables Theory}\label{sec:appendix_tables}
%

\subsection{Proof of Lemma~\ref{lem:bagArithConstraints}}
\bagArithConstraints*
\begin{proof}
  \(\varphi\) can be transformed into an equisatisfiable disjunctive normal form
  \(\varphi_1 \vee \dots \vee \varphi_n\) using standard techniques,
  where \(\varphi_i\) is a conjunction of \(\varphi_i^1, \dots, \varphi_i^{k_i}\)
  of literals for \(i \in [1,n]\).
  Each literal is either a table constraint, an arithmetic constraint,
  or an element constraint.
  For each \(i \in [1,n]\) define:
  \begin{align*}
    \Bc_i & = \{\varphi_i^j \ver \varphi_i^j \text{ is a table constraint}\}       \\
    \Ac_i & = \{\varphi_i^j \ver \varphi_i^j \text{ is an arithmetic constraint}\} \\
    \Ec_i & = \{\varphi_i^j \ver \varphi_i^j \text{ is an element constraint}\}
  \end{align*}
  where \(j \in [1, k_i]\) for each \(i \in [1,n]\).
  It is clear that \(\varphi\) is satisfiable if and only if
  \(\Bc_i \cup \Ac_i \cup \Ec_i\) is satisfiable for some \(i\in [1,n]\).
\end{proof}
\subsection{Proof of Lemma~\ref{lem:tables_eq_satisfiable}}
\tablesEqSatisfiable*
\noindent
It follows from Lemma~\ref{lem:tables_eq_satisfiable} that if all children of a configuration \(C\)
are unsatisfiable, then \(C\) is unsatisfiable.
Hence by induction on the structure of the derivation tree,
if a derivation tree of root \(c_0\) is closed,
then
\(c_0\) itself is \(\tablesTheory\)-unsatisfiable.
In the proofs, we abuse the notation \(\m{e}{s}\) to denote multiplicity term in the syntax,
and also denote multiplicity in the interpretation.
\begin{proof}
  Suppose \(c=\conf{\Ac, \Bc, \Ec}\) is satisfied by an interpretation
  \(\I\), we show that it also satisfies at least one child 
  \(c'=\conf{\Ac', \Bc', \Ec'}\) of \(C\).
  This is clear for communication rules
  \ruleBAProp, \ruleBEProp, \ruleEProp, and \ruleAProp which only add equality
  constraints satisfied by \(\I\).
  For rule \ruleBagDisequality, every model that satisfies \(s\tneq t\),
  there exists an element \(z\) in \(\I(\eleSort)\) that witnesses this disequality
  such that \(\m{z}{\I(s)} \neq \m{z}{\I(t)}\).
  \(c'\) is satisfiable by interpreting the fresh variable \(w\) as \(\I(w) = z\).
  In rules \ruleNonNegative and \ruleBagEmpty, \(c'\) is satisfiable from the semantics of \(\m{e}{s}\)
  and \(\bempty{\eleSort}\) in Section~\ref{sec:tables_theory}.
  For \ruleBagConstructorOne, the model value \(\I(\bag(e,n))\) is either
  empty or nonempty.
  If it is the former, it must be the case that \(\I(n) \leq 0\), and
  hence the left branch is satisfiable.
  The right branch is satisfied by the latter. 
  For \ruleBagConstructorTwo, if \(x \tneq e \in \Bc^*\)
  then \(\I(x \tneq e) = true\), which means \(\I(x) \neq \I(e)\) holds.
  Therefore, any child branch \(c'\) is satisfied 
  because \(\I(\m{x}{\bag(e,n)}) = 0\).
  We skip the proofs for rules \ruleDisjointUnion, \ruleMaxUnion,
  \ruleBagIntersection, \ruleDifferenceSubtract, \ruleSetof since they
  are quite similar to the cases above.
  For rule \ruleTableProductUp,
  suppose \(\I(\m{\tup{x_1, \dots, x_m}}{s}) = i\)
  and \(\I(\m{\tup{y_1, \dots, y_n}}{t}) = j\).
  Then the conclusion is satisfied by 
  \(\I(\m{\tup{x_1^\I, \dots, x_m^\I, y_1^\I, \dots, y_n^\I}}{s \product t}) = i * j \).
  The proofs of the remaining rules are similar.

  For the other direction, suppose a child configuration \(c'\)
  is satisfiable. 
  In all rules, we have
  \(\Ec \subseteq \Ec'\).
  So if \(\Ec'\) is satisfiable, then \(\Ec\) is also
  satisfiable.
  Similarly, we have \(\Ac  \subseteq \Ac'\) 
  and \(\Bc  \subseteq \Bc'\). 
  Therefore, \(C\) is satisfiable. 
\end{proof}
\begin{cor}
  \label{cor:root_is_sat}
If a leaf in a derivation tree is satisfiable, then 
the root \(c_0\) is satisfiable. 
\end{cor}
\begin{proof}
  This follows by induction from Lemma~\ref{lem:tables_eq_satisfiable}. 
\end{proof}

\subsection{Proof of Proposition~\ref{prop:refutation}}

\refutation*

\begin{proof}
   We prove by induction on the structure of the derivation tree.
   Suppose \(c=\conf{\Ac, \Bc, \Ec}\) is closed. 
   Then all leaves are \(\unsat\). 
   For the base case, suppose \(c'\) is the deepest leaf.
   Then it is \unsat because of our assumption.

   For the inductive case, suppose all children of a configuration 
   \(C\) are unsatisfiable. 
   Then \(C\) must be unsatisfiable because if it is not, 
   then from Lemma~\ref{lem:tables_eq_satisfiable} 
   one of its children is satisfiable which is a contradiction. 

   We conclude that \(C\) is \(\tablesTheory\)-unsatisfiable. 
\end{proof}

\subsection{Proof of Proposition~\ref{prop:solution}}

\solution*

%

To prove Proposition~\ref{prop:solution}, we build a model \(\I\) of \(\tablesTheory\)
that satisfies \(C\) by construction.
Then, we can argue that \(C_0\) is satisfiable by Corollary~\ref{cor:root_is_sat}.

Consider a saturated configuration $C = \conf{\Ac, \Bc, \Ec}$.
We define an interpretation \(\I\) that satisfies \(\Ac \cup \Bc \cup \Ec\).
The following definitions completely define the model $\I$.
Concretely,
\Cref{def:domains} defines the domains of $\I$;
\Cref{def:mc-fun} defines the function symbols of $\tablesSig$;
\Cref{def:mc-int-elem} defines the interpretations of variables of sort \(\int\) and \(\eleSort\);
Finally, \Cref{def:mc_bag_terms} defines the interpretations of variables of sort \(\ms\).
We start by assigning domains to the sorts of $\tablesTheory$.

\begin{definition}[Model construction: domains] \label{def:domains}
\ 

  \begin{enumerate}
    \item $\I(\sint)=\mathbb{Z}$, the set of integers.
    \item $\I(\eleSort)$ is some countably infinite set.
    \item $\I(\ms)$ is the set of finite bags whose elements are taken from $\I(\eleSort)$.
  \end{enumerate}
\end{definition}
%

The exact identity of $\eleSort$ is not important, and so we set
its elements to be arbitrary.
In contrast, its cardinality is important, and is set to be infinite.

\begin{definition}[Model construction: function symbols]
  \label{def:mc-fun}
  The symbols shown in Figure~\ref{fig:tables_sig} are interpreted in such a way that
  $\I$ is both a \(\nth\)-interpretation and a $\tablesTheory$-interpretation,
  as described in Figure~\ref{fig:tables_sig} and in Section~\ref{sec:tables_theory}.
\end{definition}

\noindent
Now we assign values to variables of sort \(\int\) and \(\eleSort\).
\begin{definition}[Model construction: $\sint$ and $\eleSort$ variables]
  \label{def:mc-int-elem}
\ 

  \begin{enumerate}
    \item\label{item:mc-arith} Let $\arithmodel$ be a \(\nth\)-interpretation with $\arithmodel(\sint)=\mathbb{Z}$
    that satisfies \(\Ac\). Then
    $\I(x):=\arithmodel(x)$ for every variable $x$ of sort $\sint$.
    \item\label{item:mc-elem} 
    We consider equivalence classes defined as follows:
    two terms \(t_1, t_2\) are in the same equivalence class if
    \(\Bc \models_{\tupleSort} t_1 \teq t_2\).
    Let \(a_1,a_2,\ldots\) be an enumeration of \(\I(\eleSort)\), and let
    \(e_1,\ldots,e_n\) be an enumeration of the equivalence classes of
    \(\bagEq{\Bc}\) whose variables have sort \(\eleSort\).
    Then, for every $i\in [1,n]$ and every variable $x\in e_i$, $\I(x):=a_i$.        
    %
    %
    %
  \end{enumerate}
\end{definition}
\noindent
We now turn to assigning values to variables of sort \(\ms\).
We interpret them as a set of pairs \(\tup{e, n}\)
where \(e \in I(\eleSort)\) and  \(n\) is a positive integer.
We use \(\cup\) to denote the union operation over sets.
Likewise, we define the disjoint union \(\uplus\) of these sets as follows:
\begin{align}
  u \uplus v & := \; \{\tup{e, n} \mid \tup{e, n} \in u, \tup{e, m} \notin v \text{ for all } m \} \; \cup \nonumber  \\
             & \qquad \{\tup{e, n} \mid \tup{e, n} \in v, \tup{e, m} \notin u \text{ for all } m \} \; \cup \nonumber \\
             & \qquad \{\tup{e, n} \mid \tup{e, n_1} \in u, \tup{e, n_2} \in v, n = n_1 + n_2  \}
  \label{eq:disjoint_union}
\end{align}
\begin{definition}[Model construction: \(\ms\) terms]
  \label{def:mc_bag_terms}
  For any variable \(s \in \vars(\Bc)\) of sort \(\ms\),
  and every \(e\) of sort \(\eleSort\),
  \begin{align}
    \I(s) & := \{\tup{\I(e), \I(\m{e}{s})} \mid {\m{e}{s} \in \ter{\Bc^*}}, \I(\m{e}{s}) > 0 \}. \label{eq:i_of_s}
  \end{align}
  There is no cycle in~\ref{eq:i_of_s} since
  \(\I(\m{e}{s})\) is an integer that has already been assigned a numeral
  in Definition~\ref{def:mc-int-elem}.
  It does not matter which element is picked 
  from its equivalence class since \(e_1 \teq e_2\) implies \(\tup{\I(e_1), \I(\m{e_1}{s})} = \tup{\I(e_2), \I(\m{e_2}{s})}\)
  from Definition~\ref{def:mc-int-elem} and~\ref{eq:b_star}. 
  %
  Equation~\ref{eq:i_of_s} is well-defined as we prove below that 
  \(s \teq t \) implies \(\I(s) = \I(t)\). 

  %
  We define the interpretation of other terms inductively as follows:
  \begin{align*}
    \I(\bempty{\eleSort}) & := \; \{\}                                                                                           \\
    \I(\bag(e, n))     & := \; \{\tup{\I(e), \I(n)} \mid \I(n) > 0 \}                                                               \\
    \I(s \sqcup t)     & := \; \{\tup{e, n} \mid \tup{e, n} \in \I(s), \tup{e, m} \notin \I(t) \text{ for all } m \} \; \cup  \\
                       & \qquad \{\tup{e, n} \mid \tup{e, n} \in \I(t), \tup{e, m} \notin \I(s) \text{ for all } m \} \; \cup \\
                       & \qquad \{\tup{e, n} \mid \tup{e, n_1} \in \I(s), \tup{e, n_2} \in \I(t), n = \max(n_1, n_2)  \}      \\
    \I(s \squplus t)   & := \; \I(s) \uplus \I(t)                                                                             \\
    \I(s \sqcap t)     & := \; \{\tup{e, n} \mid \tup{e, n_1} \in \I(s), \tup{e, n_2} \in \I(t), n = \min(n_1, n_2)  \}       \\
    \I(s \setminus t)  & := \; \{\tup{e, n} \mid \tup{e, n} \in \I(s), \tup{e, m} \notin \I(t)\text{ for all } m  \} \; \cup  \\
                       & \qquad \{\tup{e, n} \mid \tup{e, n_1} \in \I(s), \tup{e, n_2} \in \I(t), n = n_1 - n_2, n > 0  \}    \\
    \I(s \dsetminus t) & := \; \{\tup{e, n} \mid \tup{e, n} \in \I(s), \tup{e, m} \notin \I(t)\text{ for all } m \}           \\
    \I(\setof(s))      & := \; \{\tup{e, 1} \mid \tup{e, n} \in \I(s)\}                                                       \\
    \I(s \product t)   & := \; \{\tup{(e_1, e_2), n} \mid \tup{e_1, n_1} \in \I(s), \tup{e_2, n_2} \in \I(t), n = n_1 \times n_2\}         \\       
    \I(\filter(p, s))  & := \; \{\tup{e, n} \mid \tup{e, n} \in \I(s), \I(p)(e)\}                                             \\
    \I(\map(f, s))     & := \; \biguplus_{\tup{e, n} \in \I(s)}\{\tup{\I(f)(e), n}\}                                          \\
    %
  \end{align*}
where \((e_1, e_2)\) in \(\I(s \product t)\) is a shorthand for the concatenation of elements of tuple \(e_1\) and 
elements of tuple \(e_2\), and
\( e_{1_{i_k}}, e_{2_{j_k}}\) are elements at positions \(i_k, j_k\)
in tuples \(e_1, e_2\) respectively. 
The interpretation of \(s \tabjoin_{i_1j_1 \cdots i_pj_p} t\) and \(\tableProject_{i_1, \dots, i_k}(s)\) are skipped since 
the former is reduced to a filter of a product, and the latter is reduced to a map term. 
\end{definition}

\begin{proposition}[\(\I\) is a model]  \label{prop:model_proof}
  If a configuration \(C = \conf{\Ac, \Bc, \Ec}\) is saturated,
  then the interpretation \(\I\) defined in Definitions \ref{def:domains} to \ref{def:mc_bag_terms}
  satisfies all constraints in \(\Ac \cup \Bc \cup \Ec\).
\end{proposition}
\begin{proof}
  It is clear from construction of \(\I\) that it satisfies all constraints in \(\Ac \cup \Ec\).
  It also satisfies (dis)equalities in \(\Bc\) of terms with sorts \(\Ec\) and  \(\int\) since
  these constraints are in \(\Ac \cup \Ec\).
  So we focus here on (dis)equalities  in \(\Bc\) of terms with \(\ms\) sorts.
  We assume for every non-variable bag term \(t' \in \ter{\Bc^*}\),
  there exists a variable \(t\) such that \(t \teq t' \in \Bc^*\).
  This assumption is without loss of generality as any set of constraints
  can easily be transformed into an equisatisfiable set satisfying the assumption
  by the addition of fresh variables and equalities.
  %
  %
  For constraints \(s \teq t\), we prove by induction
  on the structure of \(t\) that \(\I(s) = \I(t)\).
  We show this for each case of \(t\).
  In the proofs below, we assume \(s, t, u\) are variables.
  %
  %
  \begin{enumerate}[leftmargin=*]
    \item
          \(s \teq t\) where \(t\) is a variable.
          The proof is by contradiction. 
          Suppose \(\I(s) \neq \I(t)\).
          This means there exists a pair \(\tup{\I(e), \I(\m{e}{t})} \) in \(\I(t)\)
          and not in \(\I(s)\) or vice versa.
          We prove the first case, and the proof for the second case is similar.
          Suppose \(\tup{\I(e), \I(\m{e}{t})} \in \I(t)\) and
          \(\tup{\I(e), \I(\m{e}{t})} \notin \I(s)\).
          From the definition of \(\I(t)\), we have \(\m{e}{t}\in \ter{\Bc^*}\)
          such that \(\I(\m{e}{t})  > 0\).
          Since \(s \teq t\), from the definition of \(\Bc^*\) in~\ref{eq:b_star}
          we have \(\m{e}{s}\in \ter{\Bc^*}\) and \(\m{e}{s} \teq \m{e}{t} \in \Bc^*\).
          Now the constraint \(\m{e}{s} \teq \m{e}{t} \in \Ac\), because otherwise
          \(C\) is not saturated and we can apply rule \ruleBAProp.
          Therefore, \(\I(\m{e}{s}) = \I(\m{e}{t})\).
          Now from the definition of \(\I(s)\) we have
          \(\tup{\I(e), \I(\m{e}{s})} = \tup{\I(e), \I(\m{e}{t})} \in \I(s)\)
          which contradicts \(\tup{\I(e), \I(\m{e}{t})} \notin \I(s)\).
          %

    \item
          \(s \teq \bempty{\eleSort}\).
          Suppose by contradiction that \(\I(s) \neq \I(\bempty{\eleSort}) = \{\}\).
          This means there exists a pair \(\tup{\I(e), \I(\m{e}{s})} \in \I(s)\).
          The constraint \(\m{e}{s} \teq 0 \in \Ac\) because otherwise,
          \(C\) is not saturated, and we can apply rule \ruleBagEmpty to add it to \(\Ac\).
          This means \(\I(\m{e}{s}) = 0\) and therefore,
          \(\tup{\I(e), \I(\m{e}{s})} \notin \I(s) \) which contradicts our assumption
          \(\tup{\I(e), \I(\m{e}{s})} \in \I(s)\).
          \newline
    \item
          \(s \teq \bag(e, n)\).
          Here we have two cases depending on whether \(1 \leq \I(n)\):
          \begin{enumerate}
            \item \(\I(n) \leq 0\) which means \(\I(\bag(e,n)) = \{\}\).
                  It must be the case that \(n \leq 0 \in \Ac\) and  \(s \teq \bempty{\eleSort} \in \Bc\)
                  because otherwise we can
                  apply rule \ruleBagConstructorOne to add these constraints.
                  From \(s \teq \bempty{\eleSort} \in \Bc\) we have already proven \(\I(s) = \{\}\)
                  in the previous step.
                  We can not apply the right branch of the rule because
                  it adds the constraint \(1 \leq n \) to \(\Ac\) which is not satisfied by \(\I(n) \leq 0\).
            \item \(1 \leq \I(n)\) which means \(\I(\bag(e,n)) = \{\tup{\I(e), \I(n)}\}\).
                  Suppose \(\I(s) \neq \{\tup{\I(e), \I(n)}\}\).
                  This implies either \(\tup{\I(e), \I(n)} \notin \I(s)\)
                  or there exists an element \(x\) such that \\
                  \(\tup{\I(x), \I(\m{x}{s})} \in \I(s)\)
                  and \(\tup{\I(x), \I(\m{x}{s})} \notin \{\tup{\I(e), \I(n)}\}\)
                  \begin{itemize}
                    \item Suppose \(\tup{\I(e), \I(n)} \notin \I(s)\).
                          Note that this condition also includes the case \(\tup{\I(e), m} \in \I(s), \I(n) \neq m\).
                          We have the constraints \(1 \leq n , \m{e}{s} \teq n \) in \(\Ac\) because
                          otherwise \(C\) is not saturated, and we can apply the \ruleBagConstructorOne.
                          This means \(\I(\m{e}{s}) = \I(n)\) which implies 
                          \(\tup{\I(e), \I(n)} \in \I(s)\). 
                          This contradicts our assumption.
                          We can not apply the left branch because it adds the constraint \(n \leq 0\)
                          which is not satisfied by  \(1 \leq \I(n)\).
                    \item Suppose there exists an element \(x\)
                          such that
                          \(\tup{\I(x), \I(\m{x}{s})} \in \I(s)\) and \\
                          \(\tup{\I(x), \I(\m{x}{s})} \notin \{\tup{\I(e), \I(n)}\}\).
                          From the definition of \(\I(s)\) we have \(\m{x}{s} \in \ter{\Bc^*}\),
                          and hence \(x \in \ter{\Bc^*}\).
                          We also know \(e \in \ter{\Bc^*}\) because the constraint \(s \teq \bag(e, n) \in \Bc\).
                          Now either \(x \teq e \in \Bc\) or  \(x \tneq e \in \Bc\)
                          because otherwise we can apply the splitting rule
                          \ruleEProp to \(C\).
                          Suppose \(x \teq e \in \Bc\) which means \(\I(x) = \I(e)\).
                          This implies \(\tup{\I(e), \I(\m{e}{s})} \notin \{\tup{\I(e), \I(n)}\}\).
                          This means \(\I(\m{e}{s}) \neq \I(n)\).
                          This implies
                          \(\tup{\I(e), \I(n)} \notin \I(s)\) which was proved impossible in the previous case.
                          Now suppose \(x \tneq e \in \Bc\).
                          Then, it must be the case that \(\m{x}{s}\teq 0 \in \Ac\) because otherwise we can
                          apply rule \ruleBagConstructorTwo to \(C\).
                          This implies \(\I(\m{x}{s}) = 0\) which contradicts our assumption that
                          \(\tup{\I(x), \I(\m{x}{s})} \in \I(s)\) which requires
                          \(\I(\m{x}{s}) > 0\).
                  \end{itemize}
          \end{enumerate}
          We conclude \(\I(s) = \I(\bag(e,n))\).

          %
          %
    \item
          \(s \teq \setof(t) \). We want to show
          \(\I(s) = \I(\setof(t))\).
          We know from the definition of \(\I\) that
          \(\I(\setof(t)) = \{\tup{\I(e), 1} \mid \tup{\I(e), n} \in \I(t)\}\).
          \begin{enumerate}
            \item
                  Suppose  \(\tup{\I(e), 1} \in \I(\setof(t))\).
                  %
                  Since \(t\) is a variable,  we have  \(\m{e}{t} \in \ter{\Bc^*} \) and
                  \(\I(\m{e}{t}) = n \geq 1\).
                  Now \(\Ac\) contain the constraints \(1 \leq\m{e}{t}, \m{e}{s} \teq 1 \),
                  because otherwise we can apply rule \ruleSetof.
                  The right branch can not be applied because it adds the constraint \(\m{e}{t} \leq 0 \)
                  to \(\Ac\) which is not satisfied by \(\I(\m{e}{t}) = n > 0\).
                  Therefore, \(\m{e}{s} \teq 1\) and \(\I(\m{e}{s}) = 1\).
                  Therefore, \(\tup{\I(e), 1} \in \I(s) \).

            \item
                  For the other direction suppose \(\tup{\I(e), n} \in \I(s)\).
                  From the definition of \(\I(s)\), we have \(\m{e}{s} \in \ter{\Bc^*} \) and
                  \(\I(\m{e}{s}) = n > 0\).
                  Now \(\Ac\) contains the constraints
                  \(1 \leq \m{e}{t}, \m{e}{s} \teq 1\) because otherwise we can apply rule \ruleSetof
                  and add them to \(\Ac\) using the left branch.
                  We can not apply the right branch because it adds the constraint
                  \(\m{e}{s} \teq 0\) which is not satisfied by \(\I(\m{e}{s}) = n > 0\).
                  Therefore, \(\I(\m{e}{s}) = n = 1\), and \(1 \leq \I(\m{e}{t})\).
                  Hence, \(\tup{\I(e), \I(\m{e}{t})} \in \I(t)\)
                  which means \(\tup{\I(e), 1} \in \I(\setof(t))\).
          \end{enumerate}
          We conclude \(\I(s) = \I(\setof(t))\).

    \item
          \(s \teq t \squplus u\).
          We want to show \( \I(s) = \I(t) \uplus \I(u) \).
          \begin{enumerate}
            \item%
            Suppose  \(\tup{\I(e), n} \in \I(t \squplus u) = \I(t) \uplus \I(u)\).
            %
            Since \(t, u\) are variables,
            \(\ter{\Bc^*}\) either contains \(\m{e}{t}\) or  \(\m{e}{u} \).
            Either way, \(\Ac\) contains the constraint
            \(\m{e}{s} \teq \m{e}{t} + \m{e}{u}\) because otherwise
            we can apply rule \ruleDisjointUnion and add it to \(\Ac\).
            Furthermore, the three multiplicity terms are in \(\ter{\Bc^*}\) because otherwise
            we can apply rule \ruleAProp to add the constraint to \(\Bc\).
            This means they are assigned values by the interpretation \(\I\)
            in \(\I(s)\), \(\I(t)\), and \(\I(u)\).
            Suppose  \(\I(\m{e}{t}) = n_1\), \( \I(\m{e}{u}) = n_2\)
            such that \(n = n_1 + n_2\).
            Now both \(n_1 , n_2\) are non-negative because rule \ruleNonNegative is
            saturated and \(0 \leq \m{e}{t}, 0 \leq \m{e}{u} \in \Ac\).
            Since \(n >0\) we have three cases that correspond to the three
            disjoint sets in the definition of \(\I(t) \uplus \I(u)\) in~\ref{eq:disjoint_union}:
            \((n = n_1, n_2 = 0)\),
            \((n = n_2, n_1 = 0)\) and
            \((n = n_1 + n_2, n_1, n_2 > 0)\).
            In all cases,
            we have \(\I(\m{e}{s})  = \I(\m{e}{t})+\I(\m{e}{u}) = n\).
            Hence, \(\tup{\I(e), n} \in \I(s) \).
            \item
                  Suppose  \(\tup{\I(e), n} \in \I(s)\).
                  From the definition of \(\I(s)\) we have
                  \(\m{e}{s} \in \ter{\Bc^*}\).
                  Now \(\Ac\) has the constraint
                  \(\m{e}{s} \teq \m{e}{t} + \m{e}{u}\)
                  because otherwise we can apply rule \ruleDisjointUnion
                  and add it to \(\Ac\).
                  The arithmetic solver must satisfy this constraint, and
                  we have \(\I(\m{e}{s})  = \I(\m{e}{t})+\I(\m{e}{u}) = n\).
                  Suppose  \(\I(\m{e}{t}) = n_1\),
                  \( \I(\m{e}{u}) = n_2\) such that \(n = n_1 + n_2\).
                  We know \(n_1, n_2\) are non-negative because of the saturation of
                  rule \ruleNonNegative.
                  There are 3 cases:  \((n = n_1, n_2 = 0)\),
                  \((n = n_2, n_1 = 0)\) and
                  \((n = n_1 + n_2, n_1, n_2 > 0)\).
                  In all cases we have
                  \(\tup{\I(e), n} \in \I(t \squplus u)\) from the definition of \(\I(t \squplus u)\).
          \end{enumerate}
          We conclude \(\I(s) = \I(t) \uplus \I(u)\).
          \newline
    \item  The proofs for other cases \(\sqcup, \sqcap, \setminus, \dsetminus, \product\)
          are quite similar to the proof of \(\squplus\).
          \newline
    \item
          \(s \teq \filter(p, t)\).
          Here we assume all element constraints of the form  \(p(e)\) or \(\neg p(e)\) are satisfied by \(\I\).
          \begin{enumerate}
            \item
                  Suppose  \(\tup{\I(e), n} \in \I(\filter(p, t))\).
                  From the definition of \(\I(\filter(p, t))\)
                  we have \(\tup{\I(e), n} \in \I(t) \)
                  and \(\I(p)(\I(e))\).
                  From the definition of \(\I(t)\), we have
                  \(\m{e}{t} \in \ter{\Bc^*}, \I(\m{e}{t}) > 0\).
                  Now \(\m{e}{s} \teq \m{e}{t} \in \Ac,p(e) \in \Ec\)
                  because otherwise we can apply rule \ruleBagFilterUp.
                  We can not apply the second branch because the constraint
                  \(\neg p(e)\)  is not satisfied by \(\I\).
                  Hence,  \(\I(\m{e}{s}) = \I(\m{e}{t}) = n \)
                  and \(\tup{\I(e), n} \in \I(s)\).
            \item  Suppose  \(\tup{\I(e), n} \in \I(s)\).
                  From the definition of \(\I(s)\),
                  \(\m{e}{s} \in \ter{\Bc^*}\) and \(\I(\m{e}{s}) = n > 0\).
                  Now we have \(\m{e}{t} \teq \m{e}{s}\in \Ac, p(e) \in \Ec \)
                  because otherwise we can apply rule \ruleBagFilterDown.
                  The constraints in \(\Ec\) are satisfied because \(C\) is saturated,
                  and we can not apply rule \ruleEConf.
                  Hence, \( n = \I(\m{e}{s}) = \I(\m{e}{t})\)
                  and \(\I(p)(e)\) holds.
                  Therefore, \(\tup{\I(e), n} \in \I(\filter(p, t))\).
          \end{enumerate}
          We conclude \(\I(s) = \I(\filter(p, t))\).
          \newline
    \item
          \(s \teq \map(f, t)\).
          We have two cases for \(f\):
          
          \begin{enumerate}
            \item           
           \(f\) is injective.
          We want to show \(\I(s) = \I(\map(f, t))\).
          \begin{enumerate}
            \item
                  Suppose  \(\tup{\I(e), n} \in \I(\map(f, t))\).
                  From the definition of \(\I(\map(f, t))\),
                  there exists \(x\) such that \(\I(f)(\I(x)) = \I(f(x)) = \I(e)\),
                  and \(\tup{\I(x), n} \in \I(t) \) since \(f\) is injective
                  and \(\I(\m{x}{t})= n\).
                  This implies \(\m{x}{t} \in \ter{\Bc^*}\).
                  Now \(\m{f(x)}{s} \in \ter{\Bc^*}\), and
                  \(\I(\m{x}{t}) \leq \I(\m{f(x)}{s}) \in \Ac\)
                  because otherwise we can apply rule \ruleBagMapUp.
                  Furthermore,
                  \(\m{f(x)}{s} \teq \m{w}{t} \in \Ac, f(w) \teq f(x) \in \Ec\)
                  for some fresh \(w\) because otherwise we can apply rule \ruleBagMapDownInjective.
                  Since \(f\) is injective, it must be the case that \(w \teq x \in \Bc\) because
                  otherwise we can apply the splitting rule \ruleEProp and add it to \(\Bc\).
                  We could not have applied the right branch because \(f\) is injective and
                  \(w \tneq x\) implies \(f(w) \tneq f(x)\) which would lead to a conflict.
                  Now \(w \teq x\) implies \(\I(\m{w}{t}) = \I(\m{x}{t}) = n\)
                  from~\ref{eq:b_star}.
                  This means
                  \(\I(\m{f(x)}{s}) = \I(\m{w}{t}) =  n\)
                  and hence
                  \(\tup{\I(e), n} \in \I(s)\).

            \item
                  Suppose  \(\tup{\I(e), n} \in \I(s)\).
                  From the definition of \(\I(s)\),
                  we have \(\m{e}{s} \in \ter{\Bc^*}\).
                  Now
                  \(\m{e}{s} \teq \m{w}{t} \in \Ac, f(w) \teq e \in \Ec\)
                  for some fresh \(w\) because otherwise we can apply rule \ruleBagMapDownInjective.
                  This means
                  \(\I(\m{e}{s}) = \I(\m{w}{t}) =  n\)
                  and, therefore,
                  \(\tup{\I(w), n} \in \I(t)\).
                  Since \(f\) is injective, \(\tup{\I(f)(\I(w)), n} \in \I(\map(f, t))\)
                  which means
                  \(\tup{\I(e), n} \in \I(\map(f, t))\).                  
          \end{enumerate}
          \item \(f\) is not injective. 
          Here, we assume the quantifier subsolver terminated satisfying 
          the formula below which constrains the functions 
          \(\elementIndex\) and \(\mapSum\) derived from constraint~\ref{eq:map_multiplicity}. 
          We also assume that all instantiated terms, and constraints related to \(\tablesTheory\) have been 
          added to configuration \(C\). 
          %
          \begin{align}
                   \begin{array}{ll}
             \forall \, x, i .\, & 1 \leq i \leq \delem(t) \wedge i \teq \elementIndex(x, t) \rightarrow \\
               &\quad \m{x}{t} \geq 1 \;  \\                        
               & \wedge  \left( \begin{pmatrix}
                 e \teq f(x) \;\wedge            \\
                 \mapSum(e, t, i) \teq \mapSum(e, t, i-1) + \m{x}{t}      
               \end{pmatrix} \vee
               \begin{pmatrix}
                 e \tneq f(x) \;\wedge            \\      
                 \mapSum(e, t, i) \teq \mapSum(e, t, i-1)
               \end{pmatrix} \right) \\
               & \wedge \; \forall \, y, j \; .\;  i < j  \leq \delem(t)  \wedge j \teq \elementIndex(y, t)  
                                   \rightarrow x \tneq y                
                    \end{array} \label{eq:map_reduction}
            \end{align}
          %
          \begin{enumerate}
            \item
            Suppose  \(\tup{\I(e), n} \in \I(\map(f,t))\).  
            We want to show that \(\tup{\I(e), n} \in \I(s)\).
            %
            %
            Recall \(\I(\map(f, t)) := \; \biguplus_{\tup{e, n} \in \I(t)}\{\tup{\I(f)(e), n}\}\).
            Then, there exists at least a pair \(\tup{y, n_y} \in \I(t)\)
            such that \(\I(f)(y) = \I(e)\).
            From the definition of \(\I(t)\) there exists 
            \(x\) such that \(\m{x}{t} \in \ter{\Bc^*}\),
            \(\I(x) = y\), and \(\I(\m{x}{t}) = n_y\).
            Now \(\m{f(x)}{s} \in \ter{\Bc^*}\), and
             \(\m{x}{t} \leq \m{f(x)}{s} \in \Ac\)
             because otherwise we can apply rule \ruleBagMapUp.
            The two terms \(f(x), e \in \ter{\Bc^*}\)
            are assigned the same value because 
            \(\I(f(x)) = \I(f)(\I(x)) = \I(f)(y) = \I(e)\).
            From Definition~\ref{def:mc-int-elem} they 
            should be in the same equivalence class, which 
            implies 
            %
            \(f(x) \teq e \in \Bc^*\). 
            %
            From the definition of \(\Bc^*\) in~\ref{eq:b_star}, we 
            have \(\m{f(x)}{s} \in \ter{\Bc^*}\) and \(f(x) \teq e\) which implies 
            \(\m{e}{s} \in \ter{\Bc^*}\).
            Now,
            \(\mapSum(e, t, \delem(t)) \teq \m{e}{s} \in \Ac\)
            because otherwise we can apply rule \ruleBagMapUpNonInjectiveDown.
            We prove by contradiction that \(\I(\m{e}{s}) = n \). 
            Suppose \(\I(\m{e}{s}) \neq n \), and hence, \(\I(\mapSum(e, t, \delem(t)) ) \neq n \). 
            Let
            \begin{align*}
              U = \{ &\tup{\I(x), \I(\m{x}{t})} \mid \\
              &\I(f(x)) \teq \I(e), i \teq \elementIndex(x, t), 1 \leq i \leq \delem(t),
              \m{x}{t} \in \ter{\Bc^*}\}.  
            \end{align*} 
            Then we have:
            \begin{align*}
              \I(\m{e}{s}) &= \I(\mapSum(e, t, \delem(t))) \\
              & = \sum\limits_{\tup{\I(x), \I(\m{x}{t})} \in U}{\I(\m{x}{t})}
              \neq 
              n  
              =  \sum\limits_{               
                \tup{y, n_y} \in \I(t) , 
                \I(f)(y) = \I(e)}
                n_y.
            \end{align*}            
            There are two possible reasons for the disequality: 
            \begin{enumerate}
              \item There exists element \(y\) such that 
              \(\tup{y, n_y} \in \I(t), \I(f)(y) = \I(e) , \tup{y, n_y} \not\in U\). 
              We show that this can not happen because \(C\) is saturated.
              \(\tup{y, n_y} \in \I(t)\) implies there exists \(x\) such that \(\m{x}{t} \in \ter{\Bc^*}\),
            \(\I(x) = y,\I(\m{x}{t}) = n_y\).
              Since we can not apply rule \ruleBagMapUpNonInjectiveUp,
              there exists \(1 \leq i \leq \delem(t)\)
              such that 
              \(i \teq \elementIndex(x, t)\).
              Furthermore, \(\I(f(x)) = \I(f)(\I(x)) = \I(f)(y) = \I(e)\).
              This implies  
              \(\tup{y, n_y} = \tup{\I(x), \I(\m{x}{t})} \in U\), which is a contradiction.
              \item There exists \(\tup{\I(x), \I(\m{x}{t})} \in U\) such that
              \(\I(f)(\I(x)) = \I(e)\) but \\
              \(\tup{\I(x), \I(\m{x}{t})} \not\in \I(t)\).
              This means there exists \(i\) such that 
              \(i \teq \elementIndex(x,t), 1 \leq i \leq \delem(t),
              \m{x}{t} \in \ter{\Bc^*}            
              \).
              We have
              \(\m{x}{t} \geq 1\) 
              from~\ref{eq:map_reduction},
              which implies \(\I(\m{x}{t}) > 0\).
              This means 
              \(\tup{\I(x), \I(\m{x}{t})} \in \I(t)\)
              from the definition of \(\I(t)\), a contradiction. 
            \end{enumerate}
            Therefore,  \(\I(\m{e}{s}) = n \) and hence 
            \(\tup{\I(e), n} \in \I(s)\).
            
            \item The proof for the other case is similar. 
          \end{enumerate}
        \end{enumerate}
          We conclude \(\I(s) = \I(\map(f, t))\).
  \end{enumerate}
  We now switch to disequality constraints \(s \tneq t\). 
  Without loss of generality, 
  we assume \(s, t\) are bag variables, because if they are not, 
  we can convert the constraint to an equisatisfiable one 
  that uses fresh variables \(s' \teq s \wedge t' \teq t \wedge
  s' \tneq t'\). 
   Then \(\I(s') = \I(s), \I(t') = \I(t)\).
  If \(s \tneq t \in \Bc^*\), we want to prove
  \(\I(s) \neq \I(t)\).
  For the sake of contradiction,  
  suppose \(\I(s) = \I(t)\).
  This means
  \begin{align*}
     & \{\tup{\I(e), \I(\m{e}{s})} \mid {\m{e}{s} \in \ter{\Bc^*}}, \I(\m{e}{s}) > 0 \} = \\
     & \{\tup{\I(e), \I(\m{e}{t})} \mid {\m{e}{t} \in \ter{\Bc^*}}, \I(\m{e}{t}) > 0 \}.
  \end{align*}
  Now there exists an element term \(w\) such that both \(\Ac, \Bc\) contain  the constraint
  \(\m{w}{s} \tneq \m{w}{t}\) because otherwise we can apply rule \ruleBagDisequality.
  This implies \(\I(\m{w}{s}) \neq \I(\m{w}{t})\) and at least one of them is positive
  because of the saturation of rule \ruleNonNegative.
  Without loss of generality, suppose \(\I(\m{w}{s}) > 0 \).
  This means \(\tup{\I(w),\I(\m{w}{s}) } \in \I(s)\)
  and  \(\tup{\I(w),\I(\m{w}{s}) } \notin \I(t),\)
  which contradicts our assumption.
  We conclude \(\I(s) \neq \I(t)\).
\end{proof}

\subsection{Proof of Proposition~\ref{prop:termination_no_map}}
\terminationNoMap*
\begin{proof}
  We show that any tree starts with \(C\) does not grow indefinitely. 
  Let \(C = \conf{\Ac, \Bc, \Ec}\) and 
  %
  %
  %
  \(\ter{\eleSort_c}, \ter{\int_c}, \ter{\ms_c}\)
  be all the finitely-many terms at \(C\) of sorts  \(\eleSort, \int, \ms\) respectively.
  Suppose \(e_c, i_c, b_c\) are the numbers of these terms
  respectively.
  None of the calculus rules introduces new \(\ms\) terms.
  So \(b_c\) remains the same in all derivation trees with root \(C\).
  Only rule \ruleBagDisequality generates new \(\eleSort\) terms, but it is only applied
  once for each disequality constraint.
  At the worst case all bag terms are distinct.
  So \(e_c\) can increase by at most \(b_c^2\) terms
  in all derivation trees.
  Define \(e'_c = e_c + b_c^2\).
  For multiplicity terms, assume each one of the \(e'_c\) element terms
  occurs in a multiplicity term for each bag term at the worst case.
  This means we have at most \(m_c = e'_c \times b_c\) multiplicity terms.
  So the number of integer terms would be at most \(i'_c = O(m_c^2) + i_c\)
  with \(O(m_c^2)\) accounting for terms introduced by binary arithmetic operators.
  We define the order relation \(\succ \) over configurations
  as follows:
  \begin{itemize}
    \item \(c \succ c'\) if \(c \neq \unsat\) and \(c' = \unsat\).
    \item \(c \succ c'\) if \(c \neq \unsat\) and \(c' \neq \unsat\) and
          \(
          \tup{f_1(c), \dots, f_{\terminatingBagCount}(c)} >_{lex}^{\terminatingBagCount}
          \tup{f_1(c'), \dots, f_{\terminatingBagCount}(c')}
          \)
          where \(f_i\) are the ranking functions defined in Figure~\ref{fig:bag_ranking}.
    \item   \(c \not\succ c'\) otherwise.
  \end{itemize}

  \renewcommand{\arraystretch}{1.5}
  \begin{figure}
    \small
    \begin{center}
      \begin{tabular}{l|l|l}
        \toprule
        \(f_i\)    & Rule                    & Definition                                                                          \\
        \midrule
        \(f_1\)    & \ruleBAProp             & \({i'}_c^2 - \card{\{s \teq t \mid s \teq t \in \Ac, s,t:\int \}} \)                   \\
        \(f_2\)    & \ruleBEProp             & \({e'_c}^2 - \card{\{s \teq t \mid s \teq t \in \Ec, s, t :\eleSort \}} \)             \\
        \(f_3\)    & \ruleEProp              & \(
        \begin{matrix}
          {e'_c}^2  - \card{\{s \teq t \mid s \teq t \in \Bc, s, t :\eleSort \}}  
                    - \card{\{s \tneq t \mid s \teq t \in \Bc, s, t :\eleSort \}}
        \end{matrix}
        \)                                                                                                                         \\
        \(f_4\)    & \ruleAProp              & \({i'}_c^2 - \card{\{s \teq t \mid s \teq t \in \Bc, s,t:\int \}} \) \\
        \(f_5\)    & \ruleBagDisequality     & \(b_c^2 - \card{\{\m{w}{s} \tneq \m{w}{t} \mid  \m{w}{s} \tneq \m{w}{t} \in \Ac\}} \)   \\
        \(f_6\)    & \ruleNonNegative        & \(m_c - \card{\{\m{e}{s} \mid 0 \leq \m{e}{s} \in \Ac\}} \)                         \\
        \(f_7\)    & \ruleBagEmpty           & \(e'_c \times b_c - \card{\{\m{e}{s} \mid 0 \teq \m{e}{s} \in \Ac\}} \)             \\
        \(f_8\)    & \ruleBagConstructorOne  & \(
        \begin{array}{ll}
          b_c^2 & - \card{\{s \teq \bempty{\eleSort} \mid s \teq \bempty{\eleSort} \in \Bc, s \teq \bag(e,n) \in \Bc^*\}}   \\
                & - \card{\{s \tneq \bempty{\eleSort} \mid s \tneq \bempty{\eleSort} \in \Bc, s \teq \bag(e,n) \in \Bc^*\}}
        \end{array}
        \)                                                                                                                         \\
        \(f_9\)    & \ruleBagConstructorTwo  & \( m_c
        - \card{\{\m{e}{s} \teq 0 \mid \m{e}{s} \teq 0 \in \Ac, s \teq \bag(e,n)\}} \)                                             \\
        \(f_{10}\) & \ruleDisjointUnion      & \( e'_c \times b_c^3
        - \card{\{\varphi = \m{e}{s} \teq \m{e}{t} + \m{e}{u}  \mid \varphi \in \Ac, s \teq t \squplus u \in \Bc^* \}} \)          \\
        \(f_{11}\) & \ruleMaxUnion           & \( e'_c \times b_c^3
        - \card{\{\varphi = \m{e}{s} \teq \max(\m{e}{t}, \m{e}{u})  \mid \varphi \in \Ac, s \teq t \sqcup u \in \Bc^* \}} \)       \\
        \(f_{12}\) & \ruleBagIntersection    & \( e'_c \times b_c^3
        - \card{\{\varphi = \m{e}{s} \teq \min(\m{e}{t}, \m{e}{u})  \mid \varphi \in \Ac, s \teq t \sqcap u \in \Bc^*  \}} \)      \\
        \(f_{13}\) & \ruleDifferenceSubtract & \(
        \begin{array}{ll}
          e'_c \times b_c^3 & - \card{\{\tup{\varphi_1, \varphi_2} \mid
            \varphi_1 = \m{e}{t} \leq \m{e}{u},
            \varphi_2 = \m{e}{s} \teq 0,
          \varphi_i \in \Ac, s \teq t \setminus u \in \Bc^* \}}                \\
                & - \card{\{\tup{\varphi_1, \varphi_2} \mid
            \varphi_1 = \m{e}{t} > \m{e}{u},
          \varphi_2 = \m{e}{s} \teq \m{e}{t} - \m{e}{u},                       \\
                & \qquad \varphi_i \in \Ac, s \teq t \setminus u \in \Bc^* \}}
        \end{array}
        \)                                                                                                                         \\
        \(f_{14}\) & \ruleDifferenceRemove   & \(
        \begin{array}{ll}
          e'_c \times b_c^3 & - \card{\{\tup{\varphi_1, \varphi_2} \mid
            \varphi_1 = \m{e}{u} \teq 0,
            \varphi_2 = \m{e}{s} \teq \m{e}{t},
          \varphi_i \in \Ac, s \teq t \dsetminus u \in \Bc^* \}} \\
                & - \card{\{\tup{\varphi_1, \varphi_2} \mid
            \varphi_1 = \m{e}{u} \tneq 0,
            \varphi_2 = \m{e}{s} \teq 0, \varphi_i \in \Ac, s \teq t \dsetminus u \in \Bc^* \}}
        \end{array}
        \)                                                                                                                         \\
        \(f_{15}\) & \ruleSetof              & \(
        \begin{array}{l@{~}l}
          e'_c \times b_c^2 & - \card{\{\tup{\varphi_1, \varphi_2} \mid
            \varphi_1 = 1 \leq \m{e}{t},
            \varphi_2 = \m{e}{s} \teq 1,
          \varphi_i \in \Ac, s \teq \setof(t) \in \Bc^* \}} \\
                & {} - \card{\{\tup{\varphi_1, \varphi_2} \mid
            \varphi_1 = \m{e}{t} \leq 0,
            \varphi_2 = \m{e}{s} \teq 0, \varphi_i \in \Ac, s \teq \setof(t) \in \Bc^* \}}
        \end{array}
        \)                                                                                                                         \\
        \(f_{16}\) & \ruleBagFilterUp        & \(
        \begin{array}{ll}
          e'_c \times b_c^2 & - \card{\{\tup{\varphi_1, \varphi_2} \mid
            \varphi_1 = p(e),
            \varphi_2 = \m{e}{s} \teq \m{e}{t},
          \varphi_1 \in \Ec, \varphi_2 \in \Ac, s \teq \filter(p, t) \in \Bc^* \}} \\
                & - \card{\{\tup{\varphi_1, \varphi_2} \mid
            \varphi_1 = \neg p(e),
            \varphi_2 = \m{e}{s} \teq 0,
            \varphi_1 \in \Ec, \varphi_2 \in \Ac, s \teq \filter(p, t) \in \Bc^* \}}
        \end{array}
        \)                                                                                                                         \\
        \(f_{17}\) & \ruleBagFilterDown      & \(
          e'_c \times b_c^2 - \card{\{\tup{\varphi_1, \varphi_2} \mid
          \varphi_1 = p(e),
          \varphi_2 = \m{e}{s} \teq \m{e}{t},
          \varphi_1 \in \Ec, \varphi_2 \in \Ac, s \teq \filter(p, t) \in \Bc^* \}}
        \)                                                                                                                         \\
        \bottomrule
      \end{tabular}
    \end{center}
    \caption{Ranking functions for table rules.}
    \label{fig:bag_ranking}
  \end{figure}

  Each application of any rule would only decrease its corresponding function
  and leave the other functions unchanged.
  In other words, \(f_i(c') = f_i(c) - 1\), and \(f_j(c') = f_j(c) \)
  for \(i, j \in [1, \terminatingBagCount], i \neq j\).
  It is clear that \(c \succ c'\).
  This means the derivation tree is finite, and therefore
  the calculus with the above assumptions is terminating.
\end{proof}

\section{Theory of Finite Relations}\label{sec:appendix_sets}
\begin{figure}[!htbp]
  \scriptsize
  \[
    \begin{array}{l}
      \begin{array}{@{}l@{\quad}l@{\quad}l@{\quad}l@{}}
        \toprule
        \textbf{Symbol} & \textbf{Arity}                                                             & \textbf{SMTLIB}                & \textbf{Description}                          \\
        \midrule       
        \sempty         & \s(\alpha)                                                                 & \verb|set.empty|               & \text{Empty set}                              \\
       \opsingleton{\text{-}} & \alpha \rightarrow \s(\alpha)                                   & \verb|set.singleton|                     & \text{Set constructor}                        \\        
        \sqcup          & \s(\alpha) \times \s(\alpha) \rightarrow \s(\alpha)                        & \verb|set.union|           & \text{Union}                              \\        
        \sqcap          & \s(\alpha) \times \s(\alpha) \rightarrow \s(\alpha)                        & \verb|set.inter|           & \text{Intersection}                           \\
        \setminus       & \s(\alpha) \times \s(\alpha) \rightarrow \s(\alpha)                        & \verb|set.minus| & \text{Difference}                    \\                
        \sqin           & \alpha \times \s(\alpha) \rightarrow \bool                                  & \verb|set.member|              & \text{Member }                                \\
        \sqsubseteq     & \s(\alpha) \times \s(\alpha) \rightarrow \bool                             & \verb|set.subset|              & \text{Subset}                                 \\
        \midrule
        %
        \sigma          & 
        \left(\alpha \rightarrow \bool\right) \times \s(\alpha) \rightarrow \s(\alpha)                & \verb|set.filter|              & \text{Set filter}                             \\
        \pi             & 
        \left(\alpha_1 \rightarrow \alpha_2\right) \times \s(\alpha_1) \rightarrow \s(\alpha_2)        & \verb|set.map|                 & \text{Set map}                                \\
        \midrule
        \tup{ \ldots} &
        \alpha_0 \times \dots \times \alpha_k \rightarrow \tuple(\alpha_0, \ldots, \alpha_k)                  & \verb|tuple|                    & \text{Tuple constructor}                      \\ 
        
        \select_i          & \tuple(\alpha_0, \ldots, \alpha_k) \rightarrow \alpha_i        & \verb|(_ tuple.select i)|    & \text{Tuple selector}                         
        \\[.8ex]
        \tupleProject_{i_1 \ldots i_n}        & 
        \tuple(\alpha_0, \ldots, \alpha_k) \rightarrow \tuple(\alpha_{i_1}, \ldots, \alpha_{i_n})
          & (\verb|_ tuple.proj | i_1 \cdots i_n)    & \text{Tuple projection}                         \\
        \midrule
        \product        & \r(\bm{\alpha}) \times \r(\bm{\beta}) \rightarrow \r(\bm{\alpha}, \bm{\beta}) & \verb|rel.product|           & \text{Relation cross join}                       \\
        \tabjoin_{i_1j_1 \cdots i_pj_p}      & 
          \r(\bm{\alpha}) \times \r(\bm{\beta}) \rightarrow \r(\bm{\alpha}, \bm{\beta}) 
        & (\verb|_ rel.join | i_1j_1 \cdots i_pj_p)             & \text{Relation inner join}
        \\[.8ex]        
        \relationProject_{i_1 \ldots i_n}   &
          \r(\alpha_0, \ldots, \alpha_k) \rightarrow \r(\alpha_{i_1}, \ldots, \alpha_{i_n})
            & (\verb|_ rel.proj | i_1 \cdots i_n)    & \text{Relation projection}                         \\
        \bottomrule
      \end{array}
      \\[4ex]
    \end{array}
  \]
  \caption{
    Signature \(\Sigma_r\) for the theory of relations.
    Here \(\r(\bm{\alpha}, \bm{\beta})\) is a shorthand for \(\r(\alpha_0, \ldots, \alpha_p, \beta_0, \ldots, \beta_q)\)
    when \(\bm{\alpha} = \alpha_0, \ldots, \alpha_p\) and \(\bm{\beta} = \beta_0, \ldots, \beta_q\).
  }
  \label{fig:relations_sig}
\end{figure}

\begin{figure}[!htbp]
  \begin{tabular}{c}
    \inferR[\ruleInterUp]
    {x \opin s \in \Scclosed \quad
      x \opin t \in \Scclosed  \quad
      s \opinter t \in \termsof{\Sc} }
    {\Sc := \cadd{\Sc}{x \opin s \opinter t}}
    \quad
    \inferR[\ruleInterDown]
    {x \opin s \opinter t \in \Scclosed}
    {\Sc := \cadd{\cadd{\Sc}{x \opin s}}{x \opin t}}{}
    \\[4ex]
    \inferR[\ruleUnionUp]
    {x \opin u \in \Scclosed \quad
      u \in \{s,t\}            \quad
      s \opunion t  \in \termsof{\Sc}}
    {\Sc := \cadd{\Sc}{x \opin s \opunion t}}
    \quad
    \inferR[\ruleUnionDown]
    {x \opin s \opunion t \in \Scclosed}
    {\Sc := \cadd{\Sc}{x \opin s}
      \ \parallel\
      \Sc := \cadd{\Sc}{x \opin t}}
    \\[4ex]
    \inferR[\ruleDifferenceUp]
    {x \opin s \in \Scclosed \quad
      s \opsetminus t \in \termsof{\Sc} }
    {\Sc := \cadd{\Sc}{x \opin t}
      \ \parallel\
      \Sc := \cadd{\Sc}{x \opin s \opsetminus t}}
    \quad
    \inferR[\ruleDifferenceDown]
    {x \opin s \opsetminus t \in \Scclosed}
    {\Sc := \cadd{\cadd{\Sc}{x \opin s}}{x \opnotin t}}
    \\[4ex]
    \inferR[\ruleSingleUp]
    {\opsingleton{x} \in \termsof{\Sc}}
    {\Sc := \cadd{\Sc}{x \opin \opsingleton{x}}}
    \quad
    \inferR[\ruleSingleDown]
    {x \opin \opsingleton{y} \in \Scclosed}
    {\Sc := \cadd{\Sc}{x \opequal y}}
    \quad
    \inferR[\ruleEmptyUnsat]
    {x \opin \opemptyset \in \Scclosed}
    {\unsat}
    \\[4ex]
    \inferR[\ruleSetDiseq]
    {s \not\opequal t \in \Scclosed
    }
    {\Sc := \cadd{\cadd{\Sc}{z \opin s}}{z \opnotin t}
      \quad\parallel\quad
      \Sc := \cadd{\cadd{\Sc}{z \opnotin s}}{z \opin t}}

    \quad
    \inferR[\ruleEqUnsat]
    {(t \not\opequal t) \in \Scclosed}
    {\unsat}
    \\[4ex]
    \inferR [\ruleProductUp]
    {\tup{x_1, \ldots, x_m} \sqin R_1 \in \Scclosed \quad
      \tup{y_1, \ldots, y_n} \sqin R_2 \in \Scclosed \quad
      R_1 \opprod R_2 \in \termsof{\Sc}
    }
    {\Sc := \cadd{\Sc}{\tup{x_1, \ldots, x_m, y_1, \ldots, y_n} \sqin R_1 \opprod R_2}}
    \\[4ex]
    \inferR [\ruleProductDown]
    {\tup{x_1, \ldots, x_m, y_1, \ldots, y_n} \sqin R_1 \opprod R_2 \in \Scclosed \quad
      \arity(R_1) = m}
    {\Sc := \cadd{\cadd{\Sc}{\tup{x_1, \ldots, x_m} \sqin R_1}}{\tup{y_1, \ldots, y_n} \sqin R_2}}
  \end{tabular}
  \caption{Basic rules for set intersection, union, difference, singleton, disequality and contradiction.
    In \ruleSetDiseq, $z$ is a fresh variable.
  }
  \label{fig:relations_rules}
\end{figure}

\subsection{Proofs for Relations Theory}\label{sec:sets_proofs}
\begin{restatable}{lemma}{lem:relationsProductTerminating}
  \label{lem:relations_product_terminating}
  The calculus for \(\relationsTheory\) is terminating with operators
  \(\{=, \sqcup, \sqcap, \setminus, \product, \filter\}\).
\end{restatable}
\begin{proof}
  %
  %
  Suppose we start with a configuration \(c = \conf{\Sc,\Ec}\).
  Without loss of generality assume all set terms are relations,
  and all element terms are tuples.
  Now let \(\ter{\eleSort_c}, \ter{\Sc_c}\)
  be all the finitely-many terms at \(c\) of sorts  \(\tupleSort, \s\) respectively.
  Suppose \(e_0, s_0\) are the numbers of these terms respectively.
  None of the calculus rules introduces new \(\s\) terms.
  So \(s_0\) remains constant in all derivation trees with root \(c\).
  Rule \ruleSetDiseq generates new element terms, but it is only applied
  once for each disequality constraint.
  At the worst case all set terms are distinct.
  So \(e_0\) can increase by at most \(s_0^2\) terms
  in all derivation trees using \ruleSetDiseq rule.
  Define \(e_1 = e_0 + s_0^2\).
  Rule \ruleProductUp introduces at most \(e_1^2\) new terms.
  Rule \ruleProductDown introduces at most \(2 e_1\) new element terms.
  Define \(e_2 = e_1^2 + 2e_1\).
  So the number of element terms in all derivation is at most \(e_2\).
  %
  %
  We define the order relation \(\succ \) over configurations
  as follows:
  \begin{itemize}
    \item \(c \succ c'\) if \(c \neq \unsat\) and \(c' = \unsat\).
    \item \(c \succ c'\) if \(c \neq \unsat\) and \(c' \neq \unsat\) and
          \(
          \tup{f_1(c), \dots, f_{\terminatingSetCount}(c)} >_{lex}^{\terminatingSetCount}
          \tup{f_1(c'), \dots, f_{\terminatingSetCount}(c')}
          \)
          where \(f_i\) are ranking functions defined in Figure~\ref{fig:set_ranking}.
    \item   \(c \not\succ c'\) otherwise.
  \end{itemize}

  \renewcommand{\arraystretch}{1.5}
  \begin{figure}
    \begin{center}
      \begin{tabular}{l|l|l}
        \toprule
        \(f_i\)    & Rule                  & Definition                                                                                                          \\
        \midrule
        \(f_1\)    & \ruleInterUp          & \(e_2 \times s_0^2 - \card{\{x \sqin s \sqcap t \mid x \sqin s \sqcap t \in \Sc\}} \)                               \\
        \(f_2\)    & \ruleInterDown        & \(e_2 \times s_0^2  - \card{\{x \sqin s \mid x \sqin s \in \Sc\}} \)                                                \\
        \(f_3\)    & \ruleUnionUp          & \(e_2 \times s_0^2 - \card{\{x \sqin s \sqcup t \mid x \sqin s \sqcup t \in \Sc\}} \)                               \\
        \(f_4\)    & \ruleInterDown        & \(e_2 \times s_0^2  - \card{\{x \sqin s \mid x \sqin s \in \Sc\}} \)                                                \\
        \(f_5\)    & \ruleDifferenceUp     & \(e_2 \times s_0^2  - \card{\{x \sqin s \mid x \sqin s \in \Sc\}} \)                                                \\
        \(f_6\)    & \ruleDifferenceDown   & \(e_2 \times s_0^2  - \card{\{x \not\sqin s \mid x \not\sqin s \in \Sc\}} \)                                        \\
        \(f_7\)    & \ruleSingleUp         & \(s_0  - \card{\{x \sqin [x] \mid x \sqin [x] \in \Sc\}} \)                                                         \\
        \(f_8\)    & \ruleSingleDown       & \(e_2 \times s_0  - \card{\{x \teq y \mid x \sqin [y] \in \Sc\}} \)                                                 \\
        \(f_9\)    & \ruleSetDiseq         & \(s_0^2  - \card{\{\tup{ z_{s,t} \sqin s, z_{s,t} \not\sqin t} \mid s \tneq t \in \Sc^*\}} \)                       \\
        \(f_{10}\) & \ruleProductUp        & \(e_2^2 \times s_0^2  - \card{\{\tup{ \vec{x}, \vec{y}} \mid \tup{\vec{x}, \vec{y}} \sqin s\product t \in\Sc\}} \)  \\
        \(f_{11}\) & \ruleProductDown      & \(e_2^2 \times s_0^2  - \card{\{\tup{ \vec{x}, \vec{y}} \mid \vec{x} \sqin s \in \Sc, \vec{y} \sqin t \in \Sc\}} \) \\
        \(f_{12}\) & \ruleRelationJoinUp   & \(e_2^2 \times s_0^2  - \card{\{\tup{ \vec{x}, \vec{y}} \mid \tup{\vec{x}, \vec{y}} \sqin s\product t \in\Sc\}} \)  \\
        \(f_{13}\) & \ruleRelationJoinDown & \(e_2^2 \times s_0^2  - \card{\{\tup{ \vec{x}, \vec{y}} \mid \vec{x} \sqin s \in \Sc, \vec{y} \sqin t \in \Sc\}} \) \\
        \(f_{14}\) & \ruleSetFilterUp      & \(
        \begin{array}{ll}
          e_2 \times s_0 & - \card{\{\tup{p(e), e \sqin \filter(p, s)}  \mid p(e) \in \Ec, e \sqin \filter(p, s) \in \Sc \}}                   \\
                         & - \card{\{\tup{\neg p(e), e \not\sqin \filter(p, s)}  \mid \neg p(e) \in \Ec, e \not\sqin \filter(p, s) \in \Sc \}}
        \end{array}
        \)                                                                                                                                                       \\
        \(f_{15}\) & \ruleSetFilterDown    & \(
        e_2 \times s_0 - \card{\{\tup{p(e), e \sqin s} \mid p(e) \in \Ec, e \sqin s \in \Sc\}}
        \)                                                                                                                                                       \\
        \bottomrule
      \end{tabular}
    \end{center}
    \caption{Ranking functions for relation rules.
    }
    \label{fig:set_ranking}
  \end{figure}

  Each application of any rule would only reduce its corresponding function
  and leave the rest functions unchanged.
  In other words, \(f_i(c') = f_i(c) - 1\), and \(f_j(c') = f_j(c) \)
  for \(i, j \in [1, \terminatingSetCount], i \neq j\).
  It is clear that \(c \succ c'\).
  This means the derivation tree is finite, and therefore
  the calculus with the above assumptions is terminating.
\end{proof}

\mapTermination*
\begin{proof}
  Note that the graph \(G\) does not change in all configurations
  since the calculus rules do not introduce new set terms, or new edges.
  We cut the graph \(G\) into two subgraphs \(G_1, G_2\)
  such that subgraph \(G_1\) includes all cycles in \(G\)
  and all map terms are in subgraph \(G_2\).
  %
  %
  Without loss of generality, we assume all edges in the cut set are only connected
  to map terms in \(G_2\).
  Otherwise, we would have a black edge \((u, v)\) such that \(v \in V(G_2)\) is not a map
  term, and it can be safely added to \(G_1\) because \(v\) is not in any cycle.
  From Lemma~\ref{lem:relations_product_terminating}
  the calculus is terminating in \(G_1\).
  Suppose the number of element terms in \(G_1\) after termination is \(e_1\).
  We show the calculus is terminating in \(G_2\)
  when it has no cycles.
  Suppose \(G_2\)  has \(e_2\) element terms, and \(s_2\)
  set terms.
  In the worst case, in the sense of fresh variables,
  assume all these \(s_2\) terms are map terms \(m_i \) where \(i \in [1,s_2]\)
  such that \(m_i = \map(f_i, m_{i-1}), i > 1\).
  Applying rules \ruleSetMapUp and \ruleSetMapDown would introduce at most
  \(e_2 \times s_2\) new element terms.
  Define \(e'_2 = e_1 + e_2 \times s_2\).
  Both rules \ruleSetMapUp and \ruleSetMapDown
  could be applied at most \(e'_2\) times,
  and the ranking functions would be
  \(e'_2 - \card{\{f(e) \sqin \map(f,s) \mid e \sqin s, \map(f,s) \in \ter{\Sc}\}}\)
  and
  \(e'_2 - \card{\{\tup{\varphi_1, \varphi_2} \mid
    \varphi_1 = f(w) \teq e,
    \varphi_2 = w \sqin s,
    \varphi_1 \in \Ec, \varphi_2 \in \Sc, e \sqin \map(f,s) \in \Sc^* \}}\)
  respectively.
  Therefore, the calculus is terminating in  \(G_2\).
  The calculus is still terminating if there are non-map terms in \(G_2\), 
  since the ranking functions would be similar to the ones in
  Figure~\ref{fig:set_ranking} after accounting for the new terms introduced by
  map terms.  

  Since both subgraphs \(G_1, G_2\) are terminating separately,
  we turn our attention to edges in the cut set.
  First, the \(e_2'\) elements in \(G_2\) could be propagated to \(G_1\)
  with edges of any label in  \(\{=, \sqcup, \sqcap, \setminus, \product, \filter, \map\}\).
  \(G_1\) would still be terminating with the additional \(e'_2\) elements.
  Suppose the number of elements in \(G_1\) after termination is \(e'_1\).
  %
  These 
  \(e'_1\) elements could be propagated back and forth through the edges
  in the cut set.
  However, this process is only repeated once for each edge in the cut set which is finite.
  After that, all elements would have been propagated in both \(G_1\) and \(G_2\).
  Now both subgraphs are terminating separately which concludes the proof.
\end{proof}
\section{Examples of SQL Translation}\label{sec:examples}
Figure~\ref{fig:sql_smt} shows how the theory of tables and theory of nullables 
are used to translate many SQL queries. 
Exampele~\ref{ex:sql_benchmark} shows how to translate 
one of \calcite's benchmarks into these theories. 
\begin{figure}[!htbp]
  \footnotesize
  \centering
  \begin{tabular}{l|l}
    \toprule
    SQL                                & SMT                                                                            \\
    \midrule
    int, varchar                        & \(\nullable(\int), \nullable(\stringSort)\)                                    \\
    \hline
    select * from (values ('a'),('a')) & \(\bag(\tup{"a"}, 2)\)                                                         \\
    \hline
    is null, is not null               & \(\isNull, \isSome\)                                                           \\
    \hline
    select column1 + 1 from s          & \(\map(\lambda t. \lift(+, \select_0( t), \someN(1)), s)\)                      \\
    \hline
    select * from s where column1 \(>\) 5  & \(\sigma(\lambda t. \isSome(\select_0( t)) \wedge \valN(\select_0( t)) > 5, s)\) \\
    \hline
    select * from s cross join t       & \(s \product t\)                                                               \\
    \hline
    \(\begin{array}{l}
        \text{select * from s union } \\
        \text{select * from t}
      \end{array}\)
                                       & \(s \sqcup t\)                                                                 \\
    \hline
    \(\begin{array}{l}
        \text{select * from s union all} \\
        \text{select * from t}
      \end{array}\)
                                       & \(s \squplus t\)                                                               \\
    \hline
    \(\begin{array}{l}
        \text{select * from s intersect} \\
        \text{select * from t}
      \end{array}\)
                                       & \(s \sqcap t\)                                                                 \\
    \hline
    \(\begin{array}{l}
        \text{select * from s except } \\
        \text{select * from t}
      \end{array}\)
                                       & \(s \dsetminus t\)                                                             \\
    \hline
    \(\begin{array}{l}
        \text{select * from s except all} \\
        \text{select * from t}
      \end{array}\)
                                       & \(s \setminus t\)                                                              \\
    \hline
    select distinct * from s
                                       & \(\setof(s)\)                                                                  \\
    \hline
    select * from s inner join t on p
                                       & \(\filter(p, s \product t)\)                                                   \\
    \hline
    select * from s left join t on p
                                       &
    \(\begin{array}{l}
        \filter(p, s \product t) \squplus                                                                  \\
        \map(\lambda t. \tup{\select_0( t), \dots \select_{m-1}(t), \underbrace{\nullN, \dots \nullN}_{n}}, \\
        s \dsetminus (\tableProject_{0, \dots, m-1}(\filter(p, s \product t))))
      \end{array}
    \)                                                                                                                  \\
    \hline
    select * from s right join t on p
                                       &
    \(\begin{array}{l}
        \filter(p, s \product t) \quad \squplus                                                                  \\
        \map(\lambda t. \tup{\underbrace{\nullN, \dots \nullN}_{m}, \select_0( t), \dots \select_{n-1}(t)}, \\
        t \dsetminus (\tableProject_{m, \dots, m+n-1}(\filter(p, s \product t))))
      \end{array}
    \)                                                                                                                  \\
    \hline
    select * from s full join t on p
                                       &
    \(\begin{array}{l}
        \filter(p, s \product t) \quad \squplus \\
        \map(\lambda t. \tup{\select_0( t), \dots \select_{m-1}(t), \underbrace{\nullN, \dots \nullN}_{n}}, \\
        s \dsetminus (\tableProject_{0, \dots, m-1}(\filter(p, s \product t)))) \quad \squplus                    \\                                                                \\
        \map(\lambda t. \tup{\underbrace{\nullN, \dots \nullN}_{m}, \select_0( t), \dots \select_{n-1}(t)}, \\
        t \dsetminus (\tableProject_{m, \dots, m+n-1}(\filter(p, s \product t))))
      \end{array}
    \)  \\
    \bottomrule
  \end{tabular}
  \caption{Examples of SQL translation.
  \(m, n\) are the number of columns in \(s, t\) respectively. 
  }
  \label{fig:sql_smt}
\end{figure}

\newcommand{\dept}{\textit{dept}}
\newcommand{\emp}{\textit{emp}}

\definecolor{codegreen}{rgb}{0,0.6,0}  

    \lstset{upquote=true}

    \lstdefinestyle{mystyle}{   
        basicstyle=\footnotesize\ttfamily,        
        commentstyle=\color{codegreen},
        keywordstyle=\color{blue},        
        breakatwhitespace=false,                        
        captionpos=b,                    
        keepspaces=true,                        
        showspaces=false,                
        showstringspaces=false,
        showtabs=false,      
    }
    \lstset{style=mystyle} 

\begin{example}[testPushFilterThroughSemiJoin]\label{ex:sql_benchmark}
\begin{lstlisting}[language=SQL]

q1: SELECT * FROM DEPT AS DEPT INNER JOIN (SELECT EMP.DEPTNO FROM EMP AS EMP) AS t 
ON DEPT.DEPTNO = t.DEPTNO WHERE DEPT.DEPTNO <= 10;

q2: SELECT * FROM (SELECT * FROM DEPT AS DEPT0 WHERE DEPT0.DEPTNO <= 10) AS t1 
INNER JOIN (SELECT EMP0.DEPTNO FROM EMP AS EMP0) AS t2 
ON t1.DEPTNO = t2.DEPTNO;

\end{lstlisting}
\normalsize
The two queries above are translated into theory of tables as follows:
\begin{align*}
  q1 &= \bagFilter(p_1, \bagFilter(p_0, \dept \product \tableProject_7(\emp))) \\
  q2 &= \bagFilter(p_3, \tableProject_{0,1}(\bagFilter(p_2, \dept)) \product \tableProject_7(\emp))\\
  \lift_0 & = \lift (\lambda x, y \; . \; x \teq y, \select_0(t_0), \tSelect_2(t_0)) \\
  \lift_1 & = \lift (\lambda x, y \; . \; x \leq y, \select_0(t_1), \someN(10)) \\
  \lift_2 & = \lift (\lambda x, y \; . \; x \leq y, \select_0(t_2), \someN(10)) \\
  \lift_3 & = \lift (\lambda x, y \; . \; x \teq y, \select_0(t_0), \tSelect_2(t_3)) \\
  p_i     & = \lambda t_i . \isSome(\lift_i) \wedge \valN(\lift_i), i \in \{0,1,2,3\}  
\end{align*}

Here lambda variables \(t_0, \dots, t_3\) may have different types. 
We prove the two queries are equivalent by asserting the formula \(q_1 \tneq q_2\) in \cvc which answers \unsat. 
\end{example}
\end{report}

\end{document}